\documentclass[review]{elsarticle}

\usepackage[labelfont=bf]{caption}

\usepackage{amsmath,amsfonts,amssymb,amsthm,booktabs,color,epsfig,graphicx,url,enumitem,aliascnt,multirow,array,xcolor}
\usepackage[]{hyperref}
\usepackage{natbib}

\theoremstyle{plain}% Theorem-like structures provided by amsthm.sty
\newtheorem{theorem}{Theorem}

\newaliascnt{proposition}{theorem}

\aliascntresetthe{proposition}

\newaliascnt{lemma}{theorem}

\aliascntresetthe{lemma}

\newaliascnt{corollary}{theorem}
\newtheorem{corollary}[corollary]{Corollary}
\aliascntresetthe{corollary}

\theoremstyle{definition}

\allowdisplaybreaks[4]

\numberwithin{equation}{section}

\definecolor{cyan}{RGB}{0,255,255}

\definecolor{darkgreen}{RGB}{0,100,0}
\definecolor{darkgoldenrod}{RGB}{184,134,11}
\definecolor{deepskyblue}{RGB}{0,191,255}

\newcommand{\E}{\text{E}}

\newcommand{\bI}{\mathbf{I}}
\newcommand{\bM}{\mathbf{M}}
\newcommand{\bX}{\mathbf{X}}
\newcommand{\bY}{\mathbf{Y}}
\newcommand{\bZ}{\mathbf{Z}}
\newcommand{\bA}{\mathbf{A}}
\newcommand{\bB}{\mathbf{B}}

\newcommand{\bD}{\mathbf{D}}
\newcommand{\bP}{\mathbf{P}}

\newcommand{\bS}{\mathbf{S}}
\newcommand{\bU}{\mathbf{U}}

\newcommand{\bW}{\mathbf{W}}
\newcommand{\bT}{\mathbf{T}}
\newcommand{\bx}{\mathbf{x}}

\newcommand{\bs}{\mathbf{s}}
\newcommand{\bu}{\mathbf{u}}
\newcommand{\ba}{\mathbf{a}}
\newcommand{\bb}{\mathbf{b}}

\newcommand{\bv}{\mathbf{v}}

\newcommand{\bt}{\mathbf{t}}
\newcommand{\bmu}{\boldsymbol{\mu}}

\newcommand{\balpha}{\boldsymbol{\alpha}}

\newcommand{\bSigma}{\boldsymbol{\Sigma}}
\newcommand{\bGamma}{\boldsymbol{\Gamma}}
\newcommand{\bOmega}{\boldsymbol{\Omega}}

\newcommand{\blambda}{\boldsymbol{\lambda}}

\newcommand{\bxi}{\boldsymbol{\xi}}

\newcommand{\0}{\mathbf{0}}
\newcommand{\1}{\mathbf{1}}

\begin{document}

\begin{frontmatter}

\title{Sub-dimensional Mardia measures of multivariate skewness and kurtosis}

\author[mymainaddress1]{Joydeep Chowdhury\corref{mycorrespondingauthor}}
\author[mymainaddress2]{Subhajit Dutta}
\author[mymainaddress3]{Reinaldo B. Arellano-Valle}
\author[mymainaddress1]{Marc~G.~Genton}
\address[mymainaddress1]{Statistics Program, King Abdullah University of Science and Technology, Thuwal, Saudi Arabia}
\address[mymainaddress2]{Department of Mathematics and Statistics, Indian Institute of Technology, Kanpur 208016, India}
\address[mymainaddress3]{Department of Statistics, Pontificia Universida Cat\'olica de Chile, Santiago 22, Chile}

\cortext[mycorrespondingauthor]{Corresponding author}

\begin{abstract}
The Mardia measures of multivariate skewness and kurtosis summarize the respective characteristics of a multivariate distribution with two numbers. However, these measures do not reflect the sub-dimensional features of the distribution. Consequently, testing procedures based on these measures may fail to detect skewness or kurtosis present in a sub-dimension of the multivariate distribution. We introduce sub-dimensional Mardia measures of multivariate skewness and kurtosis, and investigate the information they convey about all sub-dimensional distributions of some symmetric and skewed families of multivariate distributions. The maxima of the sub-dimensional Mardia measures of multivariate skewness and kurtosis are considered, as these reflect the maximum skewness and kurtosis present in the distribution, and also allow us to identify the sub-dimension bearing the highest skewness and kurtosis. Asymptotic distributions of the vectors of sub-dimensional Mardia measures of multivariate skewness and kurtosis are derived, based on which testing procedures for the presence of skewness and of deviation from Gaussian kurtosis are developed. The performances of these tests are compared with some existing tests in the literature on simulated and real datasets.
\end{abstract}


\end{frontmatter}

\section{Introduction}\label{sec:intro}
Consider a $ p $-variate random vector $ \bX = ( X_1, \ldots, X_p )^\top $ from a multivariate distribution with mean vector $ \bmu \in \mathbb{R}^p $ and $ p \times p $ positive definite covariance matrix $ \bSigma $. Mardia~\cite{mardia1970measures} defined the measures of multivariate skewness and kurtosis:
\begin{align}
& \beta_{1,p} = \text{E}\left[ \{ ( \bX - \bmu )^\top \bSigma^{-1} ( \bY - \bmu ) \}^3 \right]\in \Bbb{R}_+, \label{beta1p} \\
& \beta_{2,p} = \text{E}\left[ \{ ( \bX - \bmu )^\top \bSigma^{-1} ( \bX - \bmu ) \}^2 \right]\in \Bbb{R}_+, \label{beta2p}
\end{align}
respectively,
where $ \bX $ and $ \bY $ are independent and identically distributed. These measures are invariant under affine transformations and each provide a single number to summarize the skewness and the kurtosis of that $ p $-dimensional distribution.
Their sample counterparts have known asymptotic distribution and can be used to test for normality \cite{mardia1970measures,mardia1974applications}.
For the multivariate normal distribution, it is well known that $ \beta_{1,p} = 0 $ and $ \beta_{2,p} = p (p + 2) $.

One drawback of the Mardia measures of multivariate skewness and kurtosis is that they summarize the information about skewness and kurtosis too much. Sub-dimensional distributions may exhibit evidence of skewness or kurtosis, which may not be reflected in the overall Mardia measures of multivariate skewness or kurtosis.
For example, consider the celebrated Fisher's iris data for the species `iris setosa' \cite{fisher1936use}.
In \autoref{tab1}, the p-values of the Mardia test of skewness are presented for four sub-dimensions along with the whole dataset, that is, the sub-dimension $(1, 2, 3, 4)$. Here, the variables 1, 2, 3 and 4 correspond to sepal length, sepal width, petal length and petal width, respectively.
\begin{table}[t!]
\centering
\caption{p-values of the Mardia skewness test in some sub-dimensions of Fisher's iris setosa data.}
\label{tab1}
\vspace*{10pt}
\begin{tabular} {l c c c c c}  
\hline 
Sub-dimensions	&(4)	&(1, 4)	&(2, 4)	&(3, 4)	&(1, 2, 3, 4)	\\
p-values		&0.001	&0.012	&0.019	&0.018	&0.236			\\\hline
\end{tabular}
\end{table}
One can see that in the complete dataset, there is no significant evidence of skewness, while evidence of skewness in the reported sub-dimensions of dimension one (consisting of the fourth variable, namely, petal width) and of dimension two is quite strong at the 5\% level as reflected by the corresponding p-values.
This presence of skewness in the distribution of petal width of the species `iris setosa' was also observed in \cite{small1980marginal}.
Therefore, this motivates the investigation of the Mardia measures of multivariate skewness and kurtosis on sub-dimensional marginals.

Let $ \bX_{qi} $ denote a subvector of dimension $ q $ for $ 1 \leq q \leq p $ from the random vector $ \bX $, and let $ \bmu_{qi} $ and $ \bSigma_{qi} $ be the corresponding entries of $ \bmu $ and $ \bSigma $ for $ i \in \{ 1, \ldots, Q_q \} $ with $ Q_q = {p \choose q} $. We define the following sub-dimensional Mardia measures of multivariate skewness and kurtosis:
\begin{align}
& \beta_{1,q,i} = \E\left[ \{ ( \bX_{qi} - \bmu_{qi} )^\top \bSigma_{qi}^{-1} ( \bY_{qi} - \bmu_{qi} ) \}^3 \right]\in \Bbb{R}_+, \label{beta1qi} \\
& \beta_{2,q,i} = \E\left[ \{ ( \bX_{qi}-\bmu_{qi} )^\top \bSigma_{qi}^{-1} ( \bX_{qi} - \bmu_{qi} ) \}^2 \right]\in \Bbb{R}_+, \label{beta2qi}
\end{align}
for $ i \in \{1, \ldots, Q_q\} $, where $ \bX_{qi} $ and $ \bY_{qi} $ are independent and identically distributed. 
When $q=p$,  we get back $\beta_{1,p,1}\equiv \beta_{1,p}$ and $\beta_{2,p,1}\equiv \beta_{2,p}$. For $q \in \{1, \ldots, p\}$, we collect these measures in the following vectors:
\begin{align}
& \bM_{1,q} = \left( \beta_{1,q,1}, \ldots, \beta_{1,q,Q_q} \right)^\top \in \Bbb{R}_+^{Q_q}, \label{b1q} \\
& \bM_{2,q} = \left( \beta_{2,q,1}, \ldots, \beta_{2,q,Q_q} \right)^\top \in \Bbb{R}^{Q_q}_+, \label{b2q}
\end{align}
each of dimension $ Q_q = {p \choose q} $. We call \eqref{b1q} and \eqref{b2q} the $ q $-th vectors of sub-dimensional Mardia measures of multivariate skewness and kurtosis.
For the multivariate normal distribution, ${\cal N}_p(\bmu,\bSigma)$, we have:
\begin{align}
\bM_{1,q}^{\cal N} = \0_{Q_q} \quad \text{and} \quad \bM_{2,q}^{\cal N} = q (q + 2) \1_{Q_q}, \label{MardiaNormal}
\end{align}
for all $ 1 \leq q \leq p $, where $ \0_{Q_q} $ and $ \1_{Q_q} $ are $ Q_q $-dimensional vectors of zeros and of ones, respectively.

We further define:
\begin{align}
& \bM_{1}^* = \left(\bM_{1,1}^\top,\bM_{1,2}^\top, \ldots,\bM_{1,p}^\top \right)^\top=\left(\beta_{1,1,1},\ldots,\beta_{1,p,1} \right)^\top\in \Bbb{R}_+^{2^p-1}, \label{b1p} \\
& \bM_{2}^* = \left( \bM_{2,1}^\top,\bM_{2,2}^\top, \ldots,\bM_{2,p}^\top \right)^\top=\left(\beta_{2,1,1},\ldots,\beta_{2,p,1} \right)^\top\in \Bbb{R}^{2^p-1}_+. \label{b2p}
\end{align}
Here, $\bM_{1}^*$ and $\bM_{2}^*$ collect all the sub-dimensional Mardia measures of multivariate skewness and kurtosis.

\begin{figure}[t!]
\centering
\includegraphics[width=0.6\linewidth,height=7cm]{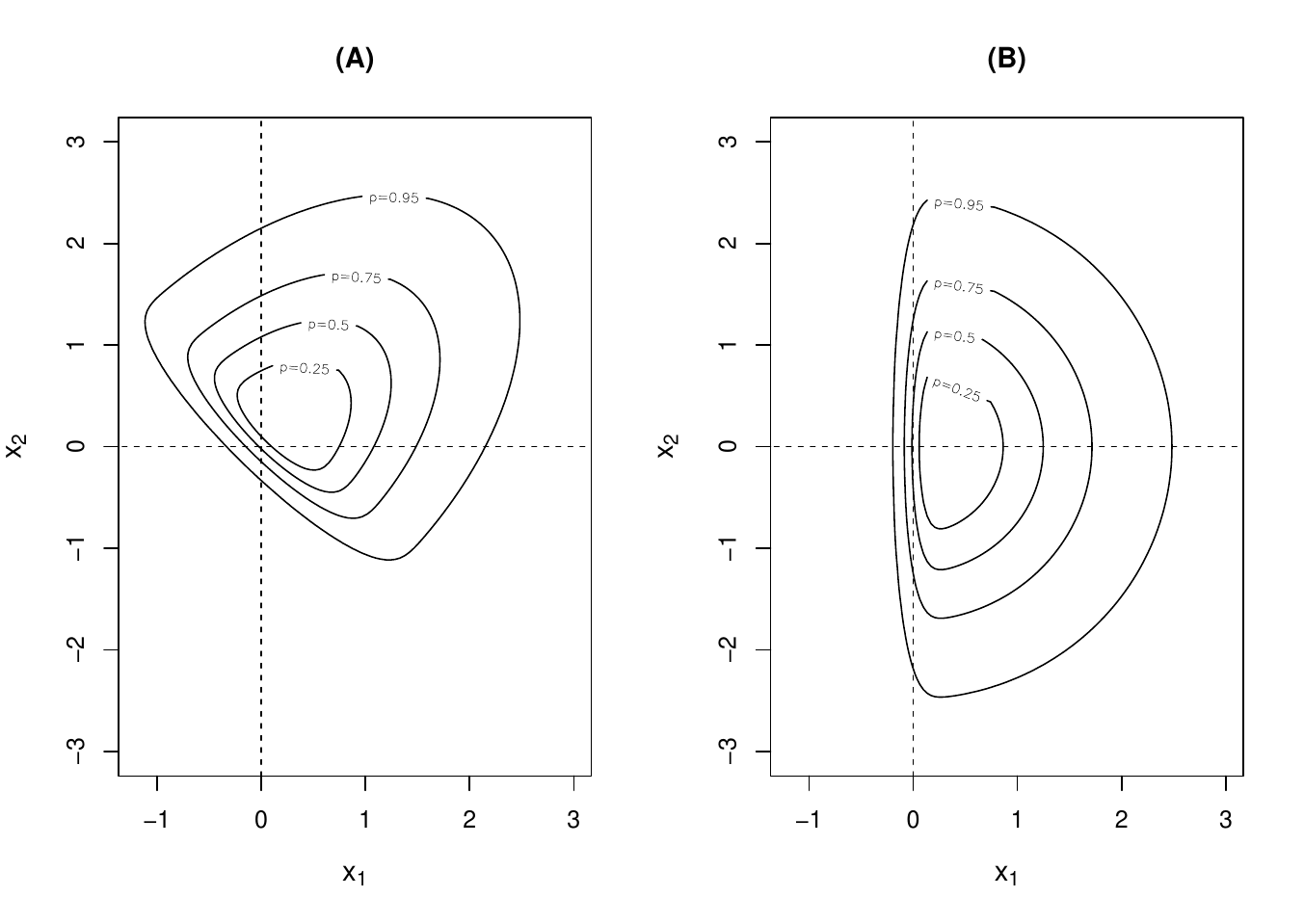}
\caption{(A) Contour plot of $ {\cal SN}_2( \mathbf{0}, \boldsymbol{\Omega}, \boldsymbol{\alpha} ) $, which has marginal skewness $ \beta_{1, 1, 1} = \beta_{1, 1, 2} = 0.130 $. (B) Contour plot of $ {\cal SN}_2( \mathbf{0}, \mathbf{I}_2, \boldsymbol{\alpha}^* ) $, which has marginal skewness $ \beta_{1, 1, 1} = 0.889 $ and $ \beta_{1, 1, 2} = 0 $. Both the skew-normal distributions have the same Mardia measure of multivariate skewness $ \beta_{1, 2} = 0.889 $.}
\label{fig:introplot}
\end{figure}
Due to affine invariance, the Mardia measures of multivariate skewness and kurtosis cannot detect any difference between distributions which are affine transformations of one another, but whose marginal skewness and/or kurtosis values are very different. In \autoref{fig:introplot}, we depict such an example for the skew-normal distribution \citep{AzzDal1996} with $p=2$. Let $ \boldsymbol{\Omega} = ( \omega_{ij} ) $, where $ \omega_{ij} = 0.5+0.5 \mathbb{I}(i = j)$ and $ \boldsymbol{\alpha} = (5, 5)^\top $. Consider the skew-normal distribution $ {\cal SN}_2( \mathbf{0}, \boldsymbol{\Omega}, \boldsymbol{\alpha} ) $ and its `canonical form' $ {\cal SN}_2( \mathbf{0}, \mathbf{I}_2, \boldsymbol{\alpha}^* ) $; see \cite[Proposition 4]{azzalini1999statistical} and \cite[subsection 5.1.8]{azzalini2013skew}. The contour plots of the densities of $ {\cal SN}_2( \mathbf{0}, \boldsymbol{\Omega}, \boldsymbol{\alpha} ) $ and $ {\cal SN}_2( \mathbf{0}, \mathbf{I}_2, \boldsymbol{\alpha}^* ) $ are presented in \autoref{fig:introplot}. From Propositions 3 and 4 in \cite{azzalini1999statistical}, it follows that the skew-normal distributions $ {\cal SN}_2( \mathbf{0}, \boldsymbol{\Omega}, \boldsymbol{\alpha} ) $ and $ {\cal SN}_2( \mathbf{0}, \mathbf{I}_2, \boldsymbol{\alpha}^* ) $ have the same values of Mardia measure of multivariate skewness, $ \beta_{1, 2}= 0.889$. However, the distribution $ {\cal SN}_2( \mathbf{0}, \mathbf{I}_2, \boldsymbol{\alpha}^* ) $ has all its skewness in its first component $ X_1 $ with the marginal distribution of $ X_2 $ being symmetric, while both the components of the distribution $ {\cal SN}_2( \mathbf{0}, \boldsymbol{\Omega}, \boldsymbol{\alpha} ) $ have the same marginal skewness. This further supports the study of Mardia measures of multivariate skewness and kurtosis on sub-dimensional marginals.

There are certain non-Gaussian distributions for which all the lower-dimensional marginals are Gaussian. An example of such a distribution, which is a trivariate generalized skew-normal distribution \citep{GenLop2005} with all the univariate and bivariate marginal distributions being standard normal, is given in \citep{loperfido2018skewness}.
Here, any procedure based on $\bM_{1,q}$ or $\bM_{2,q}$, defined in \eqref{b1q} and \eqref{b2q}, for $ q < 3 $ would not be able to detect the presence of non-Gaussianity in the distribution. However, a procedure based on $\bM_{1}^*$ defined in \eqref{b1p} would be able to detect non-Gaussianity, since $\bM_{1}^*$ and $\bM_{2}^*$ incorporate the Mardia measures on all the sub-dimensions including the whole dimension. This further motivates developing procedures based on $\bM_{1}^*$ and $\bM_{2}^*$, which are defined earlier in \eqref{b1p} and \eqref{b2p}.

In this paper, we start by studying the forms of $\bM_{1,q}$ and $\bM_{2,q}$ defined in \eqref{b1q} and \eqref{b2q} for some parametric classes of non-Gaussian distributions, including, for example, the multivariate Student's $t$ distribution and the multivariate skew-normal and skew-$t$ distributions. Then, we propose tests of normality for each dimension $q$ based on the maximum entry of $\bM_{1,q}$ and of $\bM_{2,q}$, as well as global tests of multivariate normality based on the  maximum entry of $\bM_{1}^*$ and of $\bM_{2}^*$ in \eqref{b1p} and \eqref{b2p}. One important advantage of our approach is the ability of the test, when it rejects multivariate normality, to identify the dimension $q$ and the associated sub-dimensions for which the rejection occurs.

The definitions of the Mardia measures require that the population covariance matrix $ \bSigma $ is non-singular. On the other hand, the estimation of the Mardia measures requires the sample covariance matrix to be non-singular. Estimation based on $g$-inverses is considered in \cite{mardia1975algorithm}, but the null distributions of the Mardia tests get altered when the sample covariance matrix is singular. For these reasons, the Mardia measures and associated tests are not applicable in a high-dimensional setup, particularly when $ n \le p $. While the quantities $\bM_{1}^*$ and of $\bM_{2}^*$ cannot be used because they include the Mardia measures, the quantities $\bM_{1,q}$ and $\bM_{2,q}$ can be considered for small values of $ q $ even when $ n \le p $. If it is known that the sub-dimension supporting skewness or kurtosis has a small dimension $ q $ considerably less than the sample size $ n $, then the quantities $\bM_{1,q}$ and $\bM_{2,q}$ can be employed to investigate the presence of skewness or kurtosis. Some possible ways in which the methodology can be extended to the general high-dimensional setup are discussed in \autoref{sec:conclusion}.

This paper is organized as follows. The vectors of sub-dimensional Mardia measures of multivariate skewness and kurtosis in the case of some symmetric distributions are investigated in \autoref{sec:symmetric}, whereas in the case of some skewed distributions are considered in \autoref{sec:skewed}. The invariance of these sub-dimensional measures under location-scale transformations is studied in \autoref{sec:invariance}. The new hypothesis tests are introduced in \autoref{sec:estimation} and their asymptotic distributions are established in \autoref{sec:asym}. The results of a Monte Carlo simulation study of the empirical sizes and powers of the new tests, as well as of the sub-dimensional detection, are reported in \autoref{sec:simulation}. Sub-dimensional data analyses of Fisher's iris data and of wind speed data near a wind farm in Saudi Arabia are presented in \autoref{sec:realdata}. The paper ends with a discussion in \autoref{sec:conclusion}. Some additional numerical results are provided in the supplementary material.

\section{Sub-Dimensional Mardia Measures for Some Symmetric Distributions}\label{sec:symmetric}

First, let $\bX\buildrel d\over=\bS\circ\bT\in\mathbb{R}^p$ be a standardized symmetric random vector, in which $\circ$ represents the componentwise (Hadamard) product, $\bS=\mbox{sign}(\bX)=(\mbox{sign}(X_1),\ldots,\mbox{sign}(X_p))^\top\sim {\cal U}\{-1,+1\}^p$ (discrete uniform), and $\bS$ is independent of $\bT=|\bX|=(|X_1|,\ldots,|X_p|)^\top$. Let $\bX_\pi=\bS_\pi\circ\bX=\bS_\pi\circ\bS\circ\bT$, where $\bS_\pi$ is formed by a permutation $\pi=(\pi_1,\ldots,\pi_p)$ of the components of $\bS$. Note that
$\bs\circ\bt=\bD(\bs)\bt=\bD(\bt)\bs$, where $\bD(\ba)=\text{diag}(a_1,\ldots,a_p)$. Suppose that $\bX$ has finite fourth moment with $\E(\bX)={\bf 0}_p$ and $\text{Var}(\bX)=\bI_p$. This implies $\E(\|\bX\|^2)=p$, where $\|\cdot\|$ denotes the Euclidean norm. Also, let $\bX_{\pi_*}=\bS_{\pi_*}\circ\bX$, where $\pi_*$ represents a permutation where no component remains in its original position. For example, for $p=3$, $\pi\in\{(1,2,3),(2,1,3),(2,3,1),(3,2,1),(3,1,2),(1,3,2)\}$ and $\pi_*\in\{(2,3,1), (3,1,2)\}$.  Since $\bS_{\pi_*}\circ\bS$ has mean vector ${\bf 0}_p$ and covariance matrix $\bI_p$,  we then have $\E(\bX_{\pi_*})=\E(\bS_{\pi_*}\circ\bS)\circ \E(|\bX|)={\bf 0}_p$ and $\text{Var}(\bX_{\pi_*})=\bI_p$. Also, obviously $\|\bX_{\pi_*}\|\buildrel d\over=\|\bX_\pi\|\buildrel d\over=\|\bX\|$, hence  $\E(\|\bX_{\pi_*}\|^k)=\E(\|\bX_{\pi}\|^k)=\E(\|\bX\|^k)$, $k\ge 1$. So, for the random vector $\bX_{\pi_*}$, we have
$\beta_{1,p}=0$ and $\beta_{2,p}=\E(\|\bX\|^4)$.
In particular:
\begin{enumerate}
\item If $\bX\sim {\cal N}_p({\bf 0}_p,\bI_p)$, with probability density function (pdf) given by $f(\bx)=(2\pi)^{-p/2}\exp(-\|\bx\|^2)$, $\bx\in\mathbb{R}^p$, then $\|\bX_{\pi_*}\|^2\buildrel d\over=\|\bX\|^2\sim\chi_p^2$ and $\beta_{2,p}=p(p+2)$. Therefore, for any $q$-dimensional subvector $\bX_q$ of $\bX$, we have $\bX_q\sim {\cal N}_q({\bf 0}_q,\bI_q)$ and so $\beta_{1,q}=\beta_{1,p}=0$ and $\beta_{2,q}=q(q+2)$.
Therefore, the $q$-th vectors of sub-dimensional Mardia measures of multivariate skewness and kurtosis are as in \eqref{MardiaNormal}.
\item If $\bX$ is spherically distributed with pdf  $f(\bx)=h(\|\bx\|^2)$, $\bx\in\mathbb{R}^p$, for some density generator function $h$, i.e., $h(u)>0$ for $u>0$ and $\int_0^\infty u^{p/2-1}h(u)\mbox{d}u=\Gamma(p/2)/\pi^{p/2}$, and $\text{Var}(\bX)=\sigma^2\bI_p$, then $\bX \buildrel d\over=R\bU^{(p)}$, where $R\buildrel d\over=\|\bX\|$, with  $\E(R^2)=p\sigma^2$, and $R$ is independent of $\bU^{(p)}\buildrel d\over=\bX/\|\bX\|$, the uniform vector on the $p$-dimensional unit sphere. Therefore, $\beta_{1,p}=0$ by symmetry and
$$\beta_{2,p}=\frac{p^2\E(R^4)}{\{\E(R^2)\}^2}=p(p+2)(\kappa+1),$$
where $\kappa=\gamma_2=\{\E(X_1^4)-3\}/3$ is the excess of kurtosis in the corresponding spherical univariate distribution, which can be computed from the relation:
$$\kappa+1=\frac{\beta_{2,p}}{p(p+2)}=\frac{p}{p+2}\frac{\E(R^4)}{\{\E(R^2)\}^2},$$
with the assumption $\E(R^4)<\infty$.
Moreover, any $q$-dimensional subvector $\bX_q$ of $\bX$, is also spherically distributed with stochastic representation $\bX_q=R_q\bU^{(q)}$, where $R_q=\sqrt{B_q}R$ with Beta $B_q\sim {\cal B}(q/2,(p-q)/2)$ and $R$, $B_q$ and $\bU^{(q)}$ are mutually independent. Using the fact that $\E(B_q^s)=\{\Gamma(p/2)\Gamma(q/2+s)\}/\{\Gamma(q/2)\Gamma(p/2+s)\}$, we find  again that  $\beta_{1,q}=0$ and
$$\beta_{2,q}=\frac{q^2\E(B_q^2)\E(R^4)}{\{\E(B_q)\E(R^2)\}^2}=\frac{q(q+2)}{p(p+2)}\beta_{2,p}=q(q+2)(\kappa+1).$$
Hence, for spherical distributions:
\begin{equation}
\bM_{1,q}^{\cal SPH}=\0_{Q_q}\quad \mbox{and} \quad \bM_{2,q}^{\cal SPH}=q(q+2)(\kappa+1)\1_{Q_q}, \label{MardiaSpherical}
\end{equation}
for all $1\leq q \leq p.$
For example:\\
a) If $\bX\sim {\cal N}_p({\bf 0}_p,\bI_p)$, then  $R^2\sim\chi_p^2$, with $\E(R^2)=p$, $\E(R^4)=p(p+2)$ and so
$$\kappa=\frac{p}{p+2}\frac{\E(R^4)}{\{\E(R^2)\}^2}-1=0.$$
b) If $\bX\sim t_p({\bf 0}_p,\bI_p,\nu)$, where  $t_p(\bxi,\bOmega,\nu)$ denotes the multivariate Student's $t$ distribution with location vector $\bxi$, dispersion matrix $\bOmega$, $\nu$ degrees of freedom with $\nu>4$ and pdf $c_p(\nu)|\bOmega|^{-p/2}\{1+(\bx-\bxi)^\top\bOmega^{-1}(\bx-\bxi)/\nu\}^{-(\nu+p)/2}$, $\bx\in\mathbb{R}^p$, where $c_p(\nu)=\Gamma\{(\nu+p)/2\}/\{\Gamma(\nu/2)(\nu\pi)^{p/2}\}$,  then $R^2\sim pF_{p,\nu}$ with
%$\E(R^2)=p\nu/(\nu-2)$ and $\E(R^4)=p(p+2)\nu^2/(\nu-2)(\nu-4)$
$$\E(R^{2k})=p^k\left(\frac{\nu}{p}\right)^k\frac{\Gamma(p/2+k)\Gamma(\nu/2-k)}{\Gamma(p/2)\Gamma(\nu/2)},~~\nu\ge2k,$$
thus $\kappa=2/(\nu-4)$.\\
%$X_1\sim t_1(0,1,\nu)$ with $\E(X_1^4)=3\frac{\nu-2}{\nu-4}$ and so $\kappa= 2/(\nu-4)$.\\
c) If $\bX\sim {\cal EP}_p({\bf 0}_p,\bI_p,\nu)$, where  ${\cal EP}_p(\bxi,\bOmega,\nu)$ denotes the multivariate exponential power distribution,
%with mean vector $\bmu$, covariance matrix $\bSigma$ and kurtosis parameter $\nu>0$
with location vector $\bxi$, dispersion matrix $\bOmega$ and kurtosis parameter $\nu>0$, and pdf $c_p(\nu)|\bOmega|^{-p/2}\exp[-\{(\bx-\bxi)^\top\bOmega^{-1}(\bx-\bxi)\}^\nu/2\}]$, $\bx\in\mathbb{R}^p$, where $c_p(\nu)=\{p\Gamma(p/2)\}/\{\Gamma(p/2\nu +1)2^{p/2\nu+1}\pi^{p/2}\}$,
then $R^2\buildrel d\over=V^{1/\nu}$ with Gamma $V\sim {\cal G}(p/2\nu,1/2)$. Thus, we find that
$$\E(R^{2k})=\frac{2^{k/\nu}\Gamma\{(p+2k)/(2\nu)\}}{\Gamma\{p/(2\nu)\}},~~k\ge1,$$
so
$$\kappa=\frac{p}{p+2}\frac{\Gamma\{(p+4)/(2\nu)\}\Gamma\{p/(2\nu)\}}{[\Gamma\{(p+2)/(2\nu)\}]^2}-1.$$
In this case, note that $\kappa=\kappa^{(p)}$ depends on the dimension $p$. This is due to the fact that each marginal distribution of a $p$-variate exponential power distribution depends on $p$. In particular, the univariate marginal distributions are not equivalent to the univariate one obtained by putting $p=1$. Another relevant characteristic of this distribution is that it allows both lighter $(\nu>1)$ and heavier $(0<\nu\le 1)$ tails  than the normal distribution.
\item The class of random vectors $\bX_{\pi-\pi_*}$ is also very interesting because $\E(\bX_{\pi-\pi_*})\ne{\bf 0}_p$, but some of its marginal distributions have zero mean vector. For instance, for $p=3$, a random vector in this class is given by $(S_{11}|X_1|,S_{23}|X_2|,S_{23}|X_3|)^\top,$ where $S_{ij}=S_iS_j$, in which the first component has mean $\E(|X_1|)>0$, while the remaining components have zero mean.
\end{enumerate}

\section{Sub-Dimensional Mardia Measures for Some Skewed Distributions}\label{sec:skewed}

Consider a generalized skew-normal \citep{GenLop2005} random vector $\bX\sim {\cal GSN}_p({\bf 0}_p,\bOmega,\blambda)$ with probability density function $f_{\bX}(\bx)=2\phi_p(\bx;\bOmega)G(\blambda^\top\bx)$, $\bx\in\mathbb{R}^p$ and some absolutely continuous symmetric cumulative distribution function $G$ defined on $\mathbb{R}$. Since $\beta_{1,p}$ and $\beta_{2,p}$ are invariant with respect to non-singular linear  transformations, they can be computed by using the canonical representations  $\bZ=\bGamma\bX$ and $\bZ'=\bGamma\bY$ of $\bX$ and $\bY$, respectively, in which  $\bGamma$ is a $p\times p$ orthonormal matrix such that $\bGamma\bOmega\bGamma^\top=\bI_p$ and $\bGamma\blambda=\lambda_*{\bf e}_{1:p}$, where $\lambda_*=\sqrt{\blambda^\top\bOmega\blambda}$ and ${\bf e}_1$ is the first $p$-dimensional unit vector. Then,
$Z_1\sim {\cal GSN}_1(0,1,\lambda_*)$ which has mean $\mu_*$ and variance $1-\mu_*^2$, and $Z_1$ is independent of $\bZ_{2}=(Z_2,\ldots,Z_p)^\top\sim {\cal N}_{p-1}({\bf 0}_{p-1},\bI_{p-1})$, where
$$\mu_*=2\lambda_* \E\{G'(\lambda_* Y)\}=2\frac{d \E\{G(\lambda_* Y)\}}{d\lambda_*},$$ with $Y\sim {\cal HN}_1(0,1)$ (half-normal). The same holds for $Z_1'$ and $\bZ_2'=(Z_2',\ldots,Z_p')^\top$. Note that $\bmu_\bZ=\mu_*{\bf e}_1$ and $\bSigma_{\bZ}=\bI_p-\mu_*^2{\bf e}_1{\bf e}_1^\top.$ Since
$(\bX-\bmu_\bX)^\top\bSigma_\bX^{-1}(\bX'-\bmu_\bX)=(\bZ-\bmu_\bZ)^\top\bSigma_\bZ^{-1}(\bZ'-\bmu_\bZ)=Z_{01}Z_{01}'+\bZ_2^\top\bZ_2',$
where $Z_{01}=(Z_1-\mu_*)/\sqrt{1-\mu_*^2}$ and $Z_{01}'=(Z_1'-\mu_*)/\sqrt{1-\mu_*^2}$ are standardized and independent random variables:
\begin{align*}
\beta_{1,p}&=\E\{(Z_{01}Z_{01}'+\bZ_2^\top\bZ_2')^3\}
=\E\{(Z_{01}Z_{01}')^3\}+3\E\{(Z_{01}Z_{01}')^2(\bZ_2^\top\bZ_2')\}+3\E\{(Z_{01}Z_{01}')(\bZ_2^\top\bZ_2')^2\}+\E\{(\bZ_2^\top\bZ_2')^3\}
\\
&=\E\{(Z_{01}Z_{01}')^3\}
=\gamma_{1,1}^2,
\end{align*}
where $\gamma_{1,1}\equiv\gamma_{1}=\sqrt{\beta_{1,1}}=\E(Z_{01}^3)$. Similarly, we have $(\bX-\bmu_\bX)^\top\bSigma_\bX^{-1}(\bX-\bmu_\bX)
=Z_{01}^2+\bZ_2^\top\bZ_2$ and thus:
\begin{align*}
\beta_{2,p}&=\E\{(Z_{01}^2+\bZ_2^\top\bZ_2)^2\}
=\E(Z_{01}^4)+2\E(Z_{01}^2)\E\{(\bZ_2^\top\bZ_2)\}+\E\{(\bZ_2^\top\bZ_2)^2\}
=\E(Z_{01}^4)
=\gamma_{2,1}+p(p+2),
\end{align*}
with $\gamma_{2,1}\equiv\gamma_2= \E(Z_{01}^4)-3=\beta_{2,1}-3$ being the excess of kurtosis. In this case, we can also consider the Mardia multivariate excess of kurtosis index
$\gamma_{2,p}=\beta_{2,p}-p(p+2)$ which equals $\gamma_{2,1}$.

In the skew-normal case with $G(x)=\Phi(x)$, we have $\gamma_1=b(2b^2-1)\gamma_*^{3/2}$ and $\gamma_2=2b^2(2-3b^2)\gamma_*^2$, where $\gamma_*= \lambda_*^2 / \{1+(1-b^2)\lambda_*^2\}$, $b=\sqrt{2/\pi}$ and $\lambda_* =\sqrt{\blambda^\top\bOmega\blambda}$ as defined above. However, the indices $\gamma_1$ and $\gamma_2$ corresponding to the subvectors of $\bX$ are different through subvectors with different dimensions. As was shown in \cite{azzalini1999statistical}, if $\bX\sim {\cal SN}_p(\mathbf{0}_p,\bOmega,\blambda)$ and  $\bA$ is a $q\times p$ fixed matrix, then $\bA\bX\sim {\cal SN}_q\left(\mathbf{0}_q,\bOmega_A,\blambda_A\right),$ where $\bOmega_A=\bA\bOmega\bA^\top$ and $\blambda_A=\{(\bA\bOmega\bA^\top)^{-1}\bA\bOmega\blambda\} \{1+\blambda^\top(\bOmega-\bOmega\bA^\top\bOmega_A^{-1}\bA\bOmega)\blambda\}^{-1/2}.$ In particular, if we consider the partition $\bX=(\bX_q^\top,\bX_{p-q}^\top)^\top$ with the corresponding partition for the scale matrix $\bOmega=(\bOmega_{rs})_{r,s=q,p-q}$ and the skewness vector $\blambda=(\blambda_q^\top,\blambda_{p-q}^\top)^\top$, we then have:
$$\bX_q=\bA\bX\sim {\cal SN}_q\left(\mathbf{0}_q, \bOmega_{qq},\blambda^{(q)}\right) \quad\text{with}\quad\blambda^{(q)}=\frac{\blambda_q+\bOmega_{qq}^{-1}\bOmega_{q,p-q}\blambda_{p-q}}{\sqrt{1+\blambda_{p-q}^\top\bOmega_{p-q,p-q:q}\blambda_{p-q}}},$$
where $\bA\bA^\top=\bI_q$ and $\bOmega_{p-q,p-q:q}=\bOmega_{p-q,p-q}-\bOmega_{p-q,q}\bOmega_{qq}^{-1}\bOmega_{q,p-q}$. 
%Let $$\blambda^{(q)}=\frac{\blambda_q+\bOmega_{qq}^{-1}\bOmega_{q,p-q}\blambda_{p-q}}{\sqrt{1+\blambda_{p-q}^\top\bOmega_{p-q,p-q:q}\blambda_{p-q}}}$$
%be the $q$-marginal skewness parameter.
The associated canonical transformation to $\bX_q$ has summary skewness parameter:
$$\lambda_*^{(q)}=\sqrt{\blambda^{(q)\top}\bOmega_{qq}\blambda^{(q)}}=
\sqrt{\frac{(\blambda_q+\bOmega_{qq}^{-1}\bOmega_{q,p-q}\blambda_{p-q})^\top\bOmega_{qq}(\blambda_q+\bOmega_{qq}^{-1}\bOmega_{q,p-q}\blambda_{p-q})}
{1+\blambda_{p-q}^\top\bOmega_{p-q,p-q:q}\blambda_{p-q}}}.$$

Therefore, we have
$\beta_{1,q,i}=\gamma_1^{2(q)}$ and $\beta_{2,q,i}=\gamma_2^{(q)}+q(q+2)$, where for $k=1,2$, $\gamma_k^{(q)}=\gamma_k$ for $q=p$, and for $q<p$ it  must be computed as $\gamma_k$ but with $\lambda_*$ replaced by $\lambda_*^{(q)}$. In particular, if $\bOmega=\bI_p$, then $\lambda_*^{(q)}=\sqrt{(\blambda_q^\top\blambda_q) / (1+\blambda_{p-q}^\top\blambda_{p-q})}$ with $\lambda_*^{(q)}=\lambda_*$ if $q=p$, and  $\lambda_*^{(q)}=\lambda_i/\sqrt{1+\sum_{j\ne i}\lambda_j^2}$ for the $i$-th marginal component if $q=1$.

Next, we consider the multivariate skew-$ t $ distribution as described in \cite[section 6.2]{azzalini2013skew}.
Let $ \bX \sim {\cal ST}_p( \bxi, \bOmega, \balpha, \nu) $, where $ {\cal ST}_p $ denotes a $ p $-dimensional skew-$ t $ distribution with degrees of freedom $ \nu $, location vector $ \bxi $, scale matrix $ \bOmega $ and skewness vector $ \balpha $. Define:
\begin{align*}
\delta_* = \left( \frac{\balpha^\top \bOmega \balpha}{1 + \balpha^\top \bOmega \balpha} \right)^{1/2}
\quad\text{and}\quad
b_\nu = \frac{\sqrt{\nu} \Gamma\{(\nu - 1) / 2\}}{\sqrt{\pi} \Gamma(\nu / 2)} .
\end{align*}
Based on $ \delta_* $ and $ b_\nu $, we set
$ \mu_* = b_\nu \delta_*$
and
$\sigma_*^2 = \{\nu / (\nu - 2)\} - \mu_*^2$.
The Mardia measures of multivariate skewness and kurtosis for $ \bX $ are \cite{azzalini2013skew}: 
\begin{align}
& \beta_{1,p} = \beta_1^* + 3 (p - 1) \frac{\mu_*^2}{(\nu - 3) \sigma_*^2} \; \text{if } \nu > 3 , \label{eq:skewt_skewness}\\
& \beta_{2,p} = \beta_2^* + (p^2 - 1) \frac{\nu -2}{\nu - 4} + \frac{2 (p - 1)}{\sigma_*^2} \left\{ \frac{\nu}{\nu - 4} - \frac{(\nu - 1) \mu_*^2}{\nu - 3} \right\} - p (p + 2) \; \text{if } \nu > 4 , \label{eq:skewt_kurtosis}
\end{align}
where
\begin{align*}
& \beta_1^* = \frac{\mu_*^2}{\sigma_*^3} \left\{ \frac{\nu (3 - \delta_*^2)}{\nu - 3} - \frac{3 \nu}{\nu - 2} + 2 \mu_*^2 \right\}^2 , \\
& \beta_2^* = \frac{1}{\sigma_*^4} \left\{ \frac{3 \nu^2}{(\nu - 2) (\nu - 4)} - \frac{4 \mu_*^2 \nu (3 - \delta_*^2)}{\nu - 3} + \frac{6 \mu_*^2 \nu}{\nu - 2} - 3 \mu_*^4 \right\} . 
\end{align*}

Next, denote the subvector corresponding to $ \beta_{1,q,i} $ and $ \beta_{2,q,i} $ as $ \bX_q $, and consider the partition $ \bX = ( \bX_q^\top, \bX_{p-q}^\top)^\top $ with the corresponding partitions for the location vector $ \bxi = (\bxi_q^\top, \bxi_{p-q}^\top)^\top $, scale matrix $ \bOmega = ( \bOmega_{rs} )_{r, s = q, p-q} $ and the skewness vector $ \balpha = (\balpha_q^\top, \balpha_{p-q}^\top)^\top $ as in the  earlier part of this section. Then:
\begin{align*}
\bX_q \sim {\cal ST}_q\left( \bxi_q, \bOmega_{qq}, \balpha^{(q)}, \nu \right), \quad \balpha^{(q)} = \frac{\balpha_q + \bOmega_{qq}^{-1} \bOmega_{q,p-q} \balpha_{p-q}}{\sqrt{1 + \balpha_{p-q}^\top \bOmega_{p-q,p-q : q} \balpha_{p-q}}} ,
\end{align*}
where
$\bOmega_{p-q,p-q : q} = \bOmega_{p-q,p-q} - \bOmega_{p-q,q} \bOmega_{qq}^{-1} \bOmega_{q,p-q} $. Now, $ \beta_{1,q,i} $ and $ \beta_{2,q,i} $ are obtained by replacing $ \balpha $ and $ \bOmega $ with $ \balpha^{(q)} $ and $ \bOmega_{qq} $, respectively, in the formulae given in \eqref{eq:skewt_skewness} and \eqref{eq:skewt_kurtosis}.

\section{Invariance Under Location-Scale Transformations}\label{sec:invariance}

It is well known that the Mardia measures of multivariate skewness and kurtosis \eqref{beta1p} and \eqref{beta2p} are invariant under affine transformations. We show next that the sub-dimensional Mardia measures of multivariate skewness and kurtosis \eqref{beta1qi} and \eqref{beta2qi} are only invariant
under location and scale transformation, unless the multivariate distribution is spherically invariant.

Let ${\rm MD}^2(\bX)=(\bX-\bmu)^\top \bSigma^{-1} (\bX-\bmu)$ be the squared Mahalanobis distance. Consider the affine transformation $\bY=\bA\bX+\bb$, $\bA\in\mathbb{R}^{p\times p}$ with $|\bA|\ne0$ and $\bb\in\mathbb{R}^p$. Then, it is immediate that ${\rm MD}^2(\bY)={\rm MD}^2(\bA\bX+\bb)={\rm MD}^2(\bX)={\rm MD}^2(\bZ)$, where $\bZ=\bSigma^{-1/2}(\bX-\bmu)$.

Next, partition $\bX=(\bX_q^\top,\bX_{p-q}^\top)^\top$ and similarly for $\bY$ and $\bZ$. Then, $\bX_q=\bB_q \bX$, where $\bB_q=(\bI_q,\, {\bf O})\in \mathbb{R}^{q\times p}$ with rank$(\bB_q)=q$ and ${\bf O}$ being the null matrix. Similarly, $\bY_q=\bB_q \bY$
and $\bZ_q=\bB_q \bZ$. Therefore:
\begin{eqnarray*}
{\rm MD}^2(\bX_q)&=&(\bB_q \bX-\bB_q \bmu)^\top (\bB_q \bSigma \bB_q^\top )^{-1} (\bB_q \bX-\bB_q \bmu)
= (\bX-\bmu)^\top \bB_q^\top(\bB_q \bSigma \bB_q^\top )^{-1} \bB_q(\bX-\bmu) \\
&=& \bZ^\top \bSigma^{1/2} \bB_q^\top(\bB_q \bSigma \bB_q^\top )^{-1} \bB_q  \bSigma^{1/2} \bZ
= \bZ^\top \bP \bZ,
\end{eqnarray*}
where $\bP\geq 0$ is a $p\times p$ orthogonal projection matrix, that is, symmetric with $\bP^2=\bP$ and rank$(\bP)=q$. Similarly, ${\rm MD}^2(\bY_q)=\bZ^\top \bP_\bA \bZ$ with $\bP_\bA=\bSigma^{1/2} \bA^\top \bB_q^\top(\bB_q \bA \bSigma\bA^\top \bB_q^\top )^{-1} \bB_q  \bA \bSigma^{1/2}$, where $\bP_\bA\geq 0$ is a $p\times p$ orthogonal projection matrix, that is, symmetric with $\bP_\bA^2=\bP_\bA$ and rank$(\bP_\bA)=q$. Therefore, if $\bP_\bA=\bP$ then ${\rm MD}^2(\bX_q)={\rm MD}^2(\bY_q)={\rm MD}^2(\bZ_q)=\bZ_q^\top \bZ_q$. In particular, this holds when the matrix $\bA$ is diagonal,  which proves invariance under location and scale transformations.

More generally, simple calculations show that $\bP_\bA=\bP$ if and only if:
\begin{equation}
\bA_{12}\bSigma_{21}\bA_{11}^\top+\bA_{11}\bSigma_{12}\bA_{12}^\top+\bA_{12}\bSigma_{22}\bA_{12}^\top={\bf O}.\label{invarcond}
\end{equation}
In addition to the trivial case $\bA_{12}={\bf O}$, Equation \eqref{invarcond} holds, for instance, when
$\bA_{12}\perp(\bSigma_{21}, \bSigma_{22})$, where $\perp$ means orthogonal.

Moreover, if $\bZ$ is spherically distributed, that is, $\bGamma \bX\buildrel d\over=\bX$ for any orthogonal matrix $\bGamma$, then the invariance holds for any affine transformation with $|\bA|\neq 0$. Indeed, in this case $\bP_\bA=\bGamma_\bA^\top \bD_\bA \bGamma_\bA$ with $\bD_\bA=\bB_q^\top \bB_q$ for some orthogonal matrix $\bGamma_\bA$ and therefore:
$$\bZ^\top \bP_\bA \bZ=\bZ^\top \bGamma^\top_\bA \bB_q^\top \bB_q \bGamma_\bA \bZ=(\bGamma_\bA \bZ)^\top  \bB_q^\top \bB_q(\bGamma_\bA \bZ)
\buildrel d\over=\bZ^\top  \bB_q^\top \bB_q \bZ=\bZ_q^\top \bZ_q,$$
and similarly when $\bA=\bI_p$.

In summary, although the Mardia measures of multivariate skewness and kurtosis are invariant under affine transformations, this is generally not the case for the sub-dimensional measures as shown in this section and illustrated in the simple case of a bivariate skew-normal distribution in \autoref{fig:introplot}. Therefore, the sub-dimensional Mardia measures of multivariate skewness and kurtosis are informative for testing normality in the sub-dimensions as proposed in the next section.

\section{Sub-Dimensional Estimation and Testing of Hypotheses}\label{sec:estimation}

Let $ \mathbf{X}_1, \ldots, \mathbf{X}_n $ be a random sample from a probability distribution $ F $ in $ \mathbb{R}^p $. Let $ \mathbf{X}_{qij} $ denote a sub-vector of dimension $ q $ obtained from $ \mathbf{X}_j $ for $ j \in \{1, \ldots, n\} $. The dimension of $ \mathbf{X}_{qij} $ corresponds to that of $ \mathbf{X}_{qi} $ defined in \autoref{sec:intro}, where the random vector $ \mathbf{X} $ has distribution $ F $. More precisely, suppose $ \mathbf{X}_{qi} = \mathbf{P}_{qi} \mathbf{X} $, where $ \mathbf{P}_{qi} $ is a diagonal matrix whose $ i $th diagonal element is either 1 or 0 (depending on whether the $ i $th coordinate of $ \mathbf{X} $ is included in $ \mathbf{X}_{qi} $ or not). Then, $ \mathbf{X}_{qij} = \mathbf{P}_{qi} \mathbf{X}_j $ for all $ j $.
We assume that the covariance matrix $ \bSigma $ of the distribution $ F $ is positive definite. This implies that the covariance matrices of all sub-vectors $ \mathbf{X}_{qi} $ are positive definite, and hence all of the quantities $ \beta_{1,q,i} $ and $ \beta_{2,q,i} $ defined in \eqref{beta1qi} and \eqref{beta2qi} are well defined for the underlying distribution $ F $. Further, we assume that $ n > p $ so that the sample covariances corresponding to all of the sub-vectors $ \mathbf{X}_{qi} $ are invertible almost surely. This is necessary for defining the sample estimators of $ \beta_{1,q,i} $ and $ \beta_{2,q,i} $. The dimension $ p $ is fixed in this setup.
Let $ \bar{\mathbf{X}}_{qi} $ and $ \mathbf{S}_{qi} $ be the sample mean and the sample covariance matrix of the observations $ \mathbf{X}_{qi1}, \ldots, \mathbf{X}_{qin} $, respectively. The sample estimators of $ \beta_{1,q,i} $ and $ \beta_{2,q,i} $ defined in \eqref{beta1qi} and \eqref{beta2qi} are:\\
\begin{align*}
& b_{1,q,i} = \frac{1}{n^2} \sum_{j=1}^n \sum_{k=1}^n \{(\bX_{qij}-\bar{\bX}_{qi})^\top \bS_{qi}^{-1}(\bX_{qik}-\bar{\bX}_{qi})\}^3, \\
& b_{2,q,i} = \frac{1}{n} \sum_{j=1}^n \{(\bX_{qij}-\bar{\bX}_{qi})^\top \bS_{qi}^{-1}(\bX_{qij}-\bar{\bX}_{qi})\}^2,
\end{align*}
respectively. Consequently, the estimates of $ \mathbf{M}_{1,q} $ and $ \mathbf{M}_{2,q} $ defined in \eqref{b1q} and \eqref{b2q} are
$\mathbf{m}_{1, q} = ( b_{1,q,1}, \ldots, b_{1,q,Q_q} )^\top$ and $\mathbf{m}_{2, q} = ( b_{2,q,1}, \ldots, b_{2,q,Q_q} )^\top$, respectively. Similarly, the estimates of $ \mathbf{M}_{1}^* $ and $ \mathbf{M}_{2}^* $ defined in \eqref{b1p} and \eqref{b2p} are
$\mathbf{m}_{1}^* = ( b_{1,1,1},\ldots, b_{1,p,1})^\top$ and
$\mathbf{m}_{2}^* = ( b_{2,1,1},\ldots, b_{2,p,1})^\top$, respectively.

The quantities $ \mathbf{m}_{1, q} $ and $ \mathbf{m}_{2, q} $ provide information about the skewness and kurtosis present in all the $ q $-dimensional sub-vectors constructed from the sample. On the other hand, $ \mathbf{m}_{1}^* $ and $ \mathbf{m}_{2}^* $ reflect the skewness and kurtosis present in all possible sub-dimensions of the sample. Based on these quantities, tests of skewness and kurtosis can be constructed.

\subsection{Testing Skewness}\label{subsec:skewnesstest}

From (2.4) in \cite{mardia1970measures}, it follows that $ \beta_{1,q,i} \ge 0 $ for all $ q, i $. However, for symmetric distributions, $ \beta_{1,q,i} = 0 $ for all $ q, i $ as discussed in \autoref{sec:symmetric}. So, if any of the sub-dimensions bears skewness, we shall have $ \beta_{1,q,i} > 0 $ for the $ (q, i) $-pair corresponding to that sub-dimension, which implies that $ \max\{ \beta_{1,q,i} \;|\; q \in \{1, \ldots, p\}, i \in \{1, \ldots, Q_q\} \} > 0 $. Therefore, a hypothesis for testing skewness in all sub-dimensions of the distribution can be formulated as follows:
\begin{align}
\text{H}_0^{(s)} : \max_{q, i} \beta_{1,q,i} = 0 \;\;\text{and}\;\; \text{H}_\text{A}^{(s)} : \max_{q, i} \beta_{1,q,i} > 0 .
\label{h_skewness}
\end{align}
The above hypothesis is equivalent to:
\begin{align*}
\tilde{\text{H}}_0^{(s)} : \beta_{1,q,i} = 0 \text{ for all } q, i \;\;\text{and}\;\; \tilde{\text{H}}_\text{A}^{(s)} : \beta_{1,q,i} > 0 \text{ for some } q, i.
\end{align*}
The usual Mardia skewness test \cite{mardia1970measures} only tests for:
\begin{align*}
\bar{\text{H}}_0^{(s)} : \beta_{1, p} = 0 \;\;\text{and}\;\; \bar{\text{H}}_\text{A}^{(s)} : \beta_{1, p} > 0 ,
\end{align*}
and thus it may be less efficient in providing information about skewness supported on a smaller sub-dimension. In this aspect, testing for \eqref{h_skewness} can be expected to be more efficient than the usual Mardia skewness test.

The null hypothesis $ \text{H}_0^{(s)}$ in \eqref{h_skewness} should be rejected when the maximum of $ b_{1,q,i} $ is large. However, directly comparing the $ b_{1,q,i} $s is not proper, because they have different means and standard deviations under the null hypothesis. When the underlying distribution $ F $ is Gaussian, it follows from Equation (2.26) in \cite{mardia1970measures} that the asymptotic expectation and asymptotic standard deviation of $ n b_{1,q,i} $ are $ q (q + 1) (q + 2) $ and $ \sqrt{12 q (q + 1) (q + 2)} $, respectively. So, while comparing the quantities $ b_{1,q,i} $, it is appropriate to center and scale them first. Because our aim is to detect non-Gaussian features in the sample, we center and scale $ b_{1,q,i} $ using its asymptotic expectation and standard deviation under Gaussianity, and consider:
\begin{align*}
\tilde{b}_{1,q,i} = \frac{n b_{1,q,i} - q (q + 1) (q + 2)}{\sqrt{12 q (q + 1) (q + 2)}} , \quad q \in \{1, \ldots, p\},\; i \in \{1, \ldots, Q_q\} .
\end{align*}
We reject $ \text{H}_0^{(s)}$ in \eqref{h_skewness} when $ \max_{q, i} \tilde{b}_{1,q,i} $ is large, and we shall denote this test as the MaxS test.
Let $ \tilde{\mathbf{m}}_{1, q} $ and $ \tilde{\mathbf{m}}_{1}^* $ be the centered and scaled analogues of $ \mathbf{m}_{1, q} $ and $ \mathbf{m}_{1}^* $, formed by replacing $ b_{1,q,i} $ by $ \tilde{b}_{1,q,i} $.
Note that $ \max_{q, i} \tilde{b}_{1,q,i} $ is the maximum of the coordinates of $ \tilde{\mathbf{m}}_1^* $.
To find the p-value of the MaxS test or to construct the cutoff for a given level, we need the asymptotic null distribution of $ \max_{q, i} \tilde{b}_{1,q,i} $, which is derived in \autoref{sec:asym}.

Sometimes, it may be the case that there is information about the possible presence of skewness in the sub-vectors of a fixed dimension, say, $ q_0 $, where $ q_0 < p $, but it is not known exactly which subvector has a skewed distribution. Then, it would be judicious to construct the test based on $ \max_{i} \tilde{b}_{1, q_0, i} $ only. In such a situation, the hypothesis is:
\begin{align}
\text{H}_0^{q_0, (s)} : \max_{i} \beta_{1, q_0, i} = 0 \;\;\text{and}\;\; \text{H}_\text{A}^{q_0, (s)} : \max_{i} \beta_{1, q_0, i} > 0 .
\label{h_skewness_q0}
\end{align}
Here, the null hypothesis $ \text{H}_0^{q_0, (s)} $ in \eqref{h_skewness_q0} is rejected for large values of $ \max_{i} \tilde{b}_{1, q_0, i} $. We denote this test as the $ \text{MaxS}_{q_0} $ test. To find the p-value (or, the cutoff) for rejection at a given level, we derive the asymptotic distribution of $ \max_{i} \tilde{b}_{1, q_0, i} $ in \autoref{sec:asym}.

\subsection{Testing Kurtosis}\label{subsec:kurtosistest}

Similar to tests of skewness, tests of kurtosis can be also constructed based on the $ b_{2,q,i} $ quantities. For a $ p $-dimensional Gaussian distribution, the value of the Mardia measure of multivariate kurtosis $ \beta_{2, p} $ is $ p (p + 2) $ (see \cite{mardia1970measures}). The quantity $ \beta_{2, p} $ measures how heavy-tailed the distribution is. For distributions with tails heavier than the Gaussian distribution, e.g., the multivariate Student's $ t $ distributions, we have $ \beta_{2, p} > p (p + 2) $. Similarly, for distributions with lighter tails than the Gaussian distribution, we have $ \beta_{2, p} < p (p + 2) $.
However, the Mardia test of kurtosis \cite{mardia1970measures} again checks the overall kurtosis of all the dimensions, and would not be efficient in checking whether the distribution in a particular sub-dimension deviates from Gaussianity in terms of kurtosis.

Suppose one wants to test whether the kurtosis of the distribution of any sub-dimension deviates from the Gaussian distribution. If this is the case, then we have $ \left| \beta_{2,q,i} - q (q + 2) \right| > 0 $, where the pair $ (q, i) $ corresponds to that sub-dimension. For this scenario, the suitable hypothesis is:
\begin{align}
\text{H}_0^{(k)} : \max_{q, i} \left| \beta_{2,q,i} - q (q + 2) \right| = 0 \;\;\text{and}\;\; \text{H}_\text{A}^{(k)} : \max_{q, i} \left| \beta_{2,q,i} - q (q + 2) \right| > 0 .
\label{h_kurtosis}
\end{align}
Like in the case of $ b_{1,q,i} $, $ b_{2,q,i} $ has different expectation and standard deviation for different $ q $. When the underlying distribution $ F $ is Gaussian, from the derivations in subsection 3.2 in \cite{mardia1970measures}, the asymptotic expectation and the asymptotic standard deviation of $ b_{2,q,i} $ are $ q (q + 2) $ and $ \sqrt{\{8 q (q + 2) \} / n} $. So, we center and scale $ b_{2,q,i} $, and consider:
\begin{align*}
\tilde{b}_{2,q,i} = \frac{b_{2,q,i} - q (q + 2)}{\sqrt{\{8 q (q + 2) \} / n}} , \quad q \in \{1, \ldots, p\},\; i \in \{1, \ldots, Q_q\} .
\end{align*}
The null hypothesis $ \text{H}_0^{(k)} $ in \eqref{h_kurtosis} is rejected if $ \max_{q, i} | \tilde{b}_{2,q,i}| $ is large. We denote this test as the MaxK test. To find the p-value (or the cutoff) of the test at a given level, we use the asymptotic null distribution of $ \max_{q, i} | \tilde{b}_{2,q,i}| $ derived in \autoref{sec:asym}.

Now, suppose one knows that there is possible deviation from the Gaussian distribution in terms of kurtosis in some sub-dimension with dimension $ q_0 < p $, but the exact sub-dimension is unknown. In such a case, the appropriate hypothesis is:
\begin{align}
\text{H}_0^{q_0, (k)} : \max_{i} \left| \beta_{2, q_0, i} - q_0 (q_0 + 2) \right| = 0 \;\;\text{and}\;\; \text{H}_\text{A}^{q_0, (k)} : \max_{i} \left| \beta_{2, q_0, i} - q_0 (q_0 + 2) \right| > 0 .
\label{h_kurtosis_q0}
\end{align}
Here, the null hypothesis $ \text{H}_0^{q_0, (k)} $ in \eqref{h_kurtosis_q0} is rejected if $ \max_{i} | \tilde{b}_{2,q_0,i}| $ is large. This test is denoted as the $ \text{MaxK}_{q_0} $ test. We employ the asymptotic distribution of $ \max_{i}| \tilde{b}_{2,q_0,i}| $ derived in \autoref{sec:asym} to find the p-value (or the cutoff) of the test at a given level.

\section{Asymptotic Distributions and Implementation of Tests}\label{sec:asym}

In this section, the asymptotic distributions of the quantities $ \max_{q, i} \tilde{b}_{1,q,i} $, $ \max_{i} \tilde{b}_{1, q_0, i} $, $ \max_{q, i} | \tilde{b}_{2,q,i}| $ and $ \max_{i} | \tilde{b}_{2, q_0, i}| $ introduced in \autoref{sec:estimation} are derived. Based on the respective asymptotic distributions, the implementations of the MaxS, $ \text{MaxS}_{q_0} $, MaxK and $ \text{MaxK}_{q_0} $ tests corresponding to \eqref{h_skewness}, \eqref{h_skewness_q0}, \eqref{h_kurtosis} and \eqref{h_kurtosis_q0}, respectively, are described. To derive the aforementioned asymptotic distributions, we shall use certain linearizations corresponding to the terms $ b_{1,q,i} $ and $ b_{2,q,i} $.

\subsection{Skewness}

Given any $ q, i $, define:
\begin{align}
h_{qi}( \mathbf{x}, \mathbf{y} )
&= \left\{ (\mathbf{x} - \boldsymbol{\mu}_{qi} )^\top \boldsymbol{\Sigma}_{qi}^{-1} (\mathbf{y} - \boldsymbol{\mu}_{qi}) \right\}^3 
 - 3 ( \mathbf{x} - \boldsymbol{\mu}_{qi} )^\top \boldsymbol{\Sigma}_{qi}^{-1} ( \mathbf{x} - \boldsymbol{\mu}_{qi} ) (\mathbf{x} - \boldsymbol{\mu}_{qi} )^\top \boldsymbol{\Sigma}_{qi}^{-1} (\mathbf{y} - \boldsymbol{\mu}_{qi}) \nonumber\\
& \quad - 3 ( \mathbf{y} - \boldsymbol{\mu}_{qi} )^\top \boldsymbol{\Sigma}_{qi}^{-1} ( \mathbf{y} - \boldsymbol{\mu}_{qi} ) (\mathbf{x} - \boldsymbol{\mu}_{qi} )^\top \boldsymbol{\Sigma}_{qi}^{-1} (\mathbf{y} - \boldsymbol{\mu}_{qi}) + 3 (q + 2) (\mathbf{x} - \boldsymbol{\mu}_{qi} )^\top \boldsymbol{\Sigma}_{qi}^{-1} (\mathbf{y} - \boldsymbol{\mu}_{qi}) . 
\label{eq:skewness_h}
\end{align}
Define the integral operator $ \mathbf{h}_{qi} $ by:
\begin{align}
\{\mathbf{h}_{qi}( g )\}( \mathbf{x} ) = \text{E}\left\{ h_{qi}( \mathbf{x}, \mathbf{X}_{qi} ) g( \mathbf{X}_{qi} ) \right\}.
\label{eq:integralop_h}
\end{align}
If $ \text{E}( \| \mathbf{X} \|^3 ) < \infty $, then the integral operator $ \mathbf{h}_{qi} $ is well defined on $ L_2[ \mathbf{X}_{qi} ] $, the space of measurable functions $ g $ which are square integrable with respect to the distribution of $ \mathbf{X}_{qi} $. Under appropriate assumptions, it can be established that $ \mathbf{h}_{qi} $ has only finitely many nonzero eigenvalues, where the number $ K( q ) $ of nonzero eigenvalues depends on $ q $ (see \cite{baringhaus1992limit}). Let $ \lambda_{qik} $ be the eigenvalues and $ f_{qik}( \cdot ) $ be the corresponding eigenfunctions of $ \mathbf{h}_{qi} $ for $ k \in \{1, \ldots, K( q )\} $. Then, we have:
\begin{align}
h_{qi}( \mathbf{x}, \mathbf{y} ) = \sum_{k=1}^{K( q )} \lambda_{qik} f_{qik}( \mathbf{x} ) f_{qik}( \mathbf{y} )
\label{eq:h_1}
\end{align}
for all $ \mathbf{x}, \mathbf{y} $.
The function $ h_{qi}( \mathbf{x}, \mathbf{y} ) $ is closely related to the quantity $ b_{1,q,i} $. The eigenvalues and eigenfunctions of $ h_{qi}( \mathbf{x}, \mathbf{y} ) $ are used to establish a relationship between $ b_{1,q,i} $ and the average of independent random vectors in the following theorem. This linearization will be used to derive the asymptotic distributions of the quantities $ \max_{q, i} \tilde{b}_{1,q,i} $ and $ \max_{i} \tilde{b}_{1, q_0, i} $. Define:
\begin{align*}
 m_{4, qi} &= \text{E}\left[ \left\{ ( \bX_{qi} - \boldsymbol{\mu}_{qi} )^\top \boldsymbol{\Sigma}_{qi}^{-1} ( \bX_{qi} - \boldsymbol{\mu}_{qi} ) \right\}^2 \right] , \\ m_{6, qi} &= \text{E}\left[ \left\{ ( \bX_{qi} - \boldsymbol{\mu}_{qi} )^\top \boldsymbol{\Sigma}_{qi}^{-1} ( \bX_{qi} - \boldsymbol{\mu}_{qi} ) \right\}^3 \right] .
\end{align*}
We have the following theorem on $ b_{1,q,i} $.
Note that when $q = 1$, $\mathbf{X}_{qi} = X_i$, the $i$-th component of $\mathbf{X}$ for $i \in \{1, \ldots, p\}$. Let $\mu_i$ be the $i$-th component of $\boldsymbol{\mu}$ and $X_{ij}$ be the $i$-th component of the observation vector $\mathbf{X}_j$ for $i \in \{1, \ldots, p\}$ and $j \in \{1, \ldots, n\}$.
\begin{theorem}\label{thm:skewness1}
Suppose that $ \text{E}( \| \mathbf{X} \|^6 ) < \infty $ and the distribution of $ \mathbf{X} $ is elliptical.
Let $ \sigma_i^{2} = \text{E}\{( X_i - \mu_i )^2 \} $ for $i \in \{1, \ldots, p\}$.
For $ j \in \{1, \ldots, n\} $, define $ \bu_{qij} $ as:
\begin{align*}
\bu_{qij} = 
\begin{cases}
& \sigma_i^{-3} ( X_{ij} - \mu_{i}) \{( X_{ij} - \mu_{i})^2 - 3 \sigma_i^{2}\} \text{ if } q = 1 , \\
& \left( \sqrt{\lambda_{qi1}} f_{qi1}( \mathbf{X}_{qij} ), \ldots, \sqrt{\lambda_{qiK(q)}} f_{qiK( q )}( \mathbf{X}_{qij} )\right)^\top \text{ otherwise}.
\end{cases}
\end{align*}
Let $ \bar{\bu}_{qi} = n^{-1} \sum_{j=1}^{n} \bu_{qij} $. Then:
\begin{align*}
n b_{1,q,i} 
&= \left\| \sqrt{n} \bar{\bu}_{qi} \right\|_2^2 + o_P( 1 )
% \\
%&=
%\begin{cases}
%& \left( n^{-1/2} \sum_{j=1}^{n} \sigma_i^{-3} \left( \bX_{qij} - \bmu_{qi} \right) \left[ \left( \bX_{qij} - \bmu_{qi} \right)^2 - 3 \sigma_i^{2} \right] \right)^2 + o_P( 1 ) \text{ if } q = 1 , \\
%& \sum_{k=1}^{K(q)} \left( n^{-1/2} \sum_{j=1}^{n} \sqrt{\lambda_{qik}} f_{qik}( \mathbf{X}_{qij} ) \right)^2 + o_P( 1 ) \text{ otherwise},
%\end{cases}
\end{align*}
as $ n \to \infty $. Further, for $ q > 1 $, $ K(q) = q + q (q - 1) (q + 4) / 6 $:
\begin{align*}
\lambda_{qik} = 
\begin{cases}
& (3 / q) \left\{ m_{6, qi} / (q + 2) - 2 m_{4, qi} + (q + 2) q \right\} \text{ for } k \in \{1, \ldots, q\}, \\
& 6 m_{6, qi} / \left\{ q (q + 2) (q + 4) \right\} \text{ for } k \in \{(q + 1), \ldots, K(q)\},
\end{cases}
\end{align*}
and $ \text{E}\{ f_{qik}( \bX_{qij} ) \} = 0 $ for all $ q, i, j $ and $ k $.
\end{theorem}
\begin{proof}[\textbf{\upshape Proof:}]
When $ q = 1 $, from the arguments in the proof of Theorem 1 in \cite{gastwirth1977classical}, it follows that, as $ n \to \infty $:
\begin{align*}
n b_{1,q,i} &= \left[ n^{-1/2} \sum_{j=1}^{n} \sigma_i^{-3}( X_{ij} - \mu_{i} ) \{ ( X_{ij} - \mu_{i})^2 - 3 \sigma_i^{2} \} \right]^2 + o_P( 1 )= \left( \sqrt{n} \bar{\bu}_{qi} \right)^2 + o_P( 1 ).
\end{align*}

For $ q > 1 $, let $ \bY_{qij} = \boldsymbol{\Sigma}_{qi}^{-1/2} \left( \bX_{qij} - \bmu_{qi} \right) $ for all $ i, j $. Then $ \bY_{qij} $ has an elliptical distribution, which is identical over $ j $ for fixed $ i $, with $ \text{E}( \bY_{qij} ) = \0 $ and $ \text{E}( \bY_{qij} \bY_{qij}^\top ) = \bI_q $. Also, $ \text{E}( \| \bY_{qij} \|^6 ) < \infty $ for all $ i, j $, and hence $ \bY_{qij} $s satisfy the conditions of Lemma 2.1 in \cite{baringhaus1992limit}. From an application of this lemma:
\begin{align}
n b_{1,q,i} = n^{-1} \sum_{j=1}^{n} \sum_{k=1}^{n} h_{qi}( \bX_{qij}, \bX_{qik} ) + o_P( 1 )
\label{eq:h_2}
\end{align}
as $ n \to \infty $.
Next, from the arguments in the proof of Theorem 2.2 in \cite{baringhaus1992limit}, it follows that the integral operator $ \mathbf{h}_{qi} $ defined in \eqref{eq:integralop_h} has only two non-zero distinct eigenvalues, which are:
\begin{align*}
& \gamma_{qi1} = (3 / q) \left\{ m_{6, qi} / (q + 2) - 2 m_{4, qi} + (q + 2) q \right\} ,\\
& \gamma_{qi2} = 6 m_{6, qi} / \left\{ q (q + 2) (q + 4) \right\},
\end{align*}
with associated multiplicities $ \nu_{qi1} = q $ and $ \nu_{qi2} = q (q - 1) (q + 4) / 6 $, respectively. So, we can take $ \lambda_{qik} = \gamma_{qi1} $ for $ k \in \{1, \ldots, \nu_{qi1}\} $ and $ \lambda_{qik} = \gamma_{qi2} $ for $ k \in \{ (\nu_{qi1} + 1), \ldots, (\nu_{qi1} + \nu_{qi2}) \} $. Consequently, $ K(q) = q + q (q - 1) (q + 4) / 6 $. From \eqref{eq:h_1} and \eqref{eq:h_2}, we get:
\begin{align*}
n b_{1,q,i} &= n^{-1} \sum_{j=1}^{n} \sum_{l=1}^{n} \sum_{k=1}^{K(q)} \lambda_{qik} f_{qik}( \bX_{qij} ) f_{qik}( \bX_{qil} ) + o_P( 1 ) = \sum_{k=1}^{K(q)} \left\{ n^{-1/2} \sum_{j=1}^{n} \sqrt{\lambda_{qik}} f_{qik}( \bX_{qij} ) \right\}^2 + o_P( 1 ) \\
&= \left\| \sqrt{n} \bar{\bu}_{qi} \right\|_2^2 + o_P( 1 )
\end{align*}
as $ n \to \infty $.
Finally, from the arguments in the proof of Theorem 2.1 in \cite{gregory1977large}, it follows that $ \text{E}\{ f_{qik}( \bX_{qij} ) \} = 0 $ for all $ q, i, j $ and $ k $.
\end{proof}

Define $ \bU_j = \left( \bu_{1 1 j}^\top, \ldots, \bu_{q i j}^\top, \ldots, \bu_{p 1 j}^\top \right)^\top $ and $ \bU_{q j} = \left( \bu_{q 1 j}^\top, \ldots, \bu_{q Q_q j}^\top \right)^\top $.
Let $ G( \cdot ) $ and $ G_q( \cdot ) $ be such that:
\begin{align*}
G( \bU_j ) = \max\left\{ \frac{\left\| \bu_{1 1 j} \right\|_2^2 - 6}{\sqrt{72}}, \ldots, \frac{\left\| \bu_{q i j} \right\|_2^2 - q (q + 1) (q + 2)}{\sqrt{12 q (q + 1) (q + 2)}}, \ldots, \frac{\left\| \bu_{p 1 j} \right\|_2^2 - p (p + 1) (p + 2)}{\sqrt{12 p (p + 1) (p + 2)}} \right\},
\end{align*}
and
\begin{align*}
G_q( \bU_{qj} ) = \max\left\{ \frac{\left\| \bu_{q 1 j} \right\|_2^2 - q (q + 1) (q + 2)}{\sqrt{12 q (q + 1) (q + 2)}}, \ldots, \frac{\left\| \bu_{q Q_q j} \right\|_2^2 - q (q + 1) (q + 2)}{\sqrt{12 q (q + 1) (q + 2)}} \right\} .
\end{align*}
Clearly, $ G( \cdot ) $ and $ G_q( \cdot ) $ are continuous functions. From these observations, we derive the asymptotic null distributions of the test statistics $ \max_{q, i} \tilde{b}_{1,q,i} $ for the MaxS test and $ \max_{i} \tilde{b}_{1, q_0, i} $ for the $ \text{MaxS}_{q_0} $ test in the next theorem.
\begin{theorem}\label{thm:skewness2}
Let $ \bOmega $ and $ \bOmega_q $ be the dispersion matrices of $ \bU_1 $ and $ \bU_{q1} $, respectively. Let $ \bW $ and $ \bW_{q_0} $ be zero-mean Gaussian random vectors with dispersion matrices $ \bOmega $ and $ \bOmega_{q_0} $, respectively. Assume $ \text{E}( \| \mathbf{X} \|^6 ) < \infty $ and the distribution of $ \mathbf{X} $ is elliptical.
Then,
$ \max_{q, i} n \tilde{b}_{1,q,i} \stackrel{d}{\rightarrow} G( \bW ) $ and
$ \max_{i} n \tilde{b}_{1, q_0, i} \stackrel{d}{\rightarrow} G_{q_0}( \bW_{q_0} )$
as $ n \to \infty $.
\end{theorem}
\begin{proof}[\textbf{\upshape Proof:}]
Since we have $ \text{E}\{ f_{qik}( \bX_{qij} ) \} = 0 $ for all $ q, i, j $ and $ k $ from \autoref{thm:skewness1}, it follows that $ \text{E}( \bU_{qj} ) = \bf0 $ and $ \text{E}( \bU_j ) = \bf0 $ for all $ j $. Further, the distributions of $ \bU_{qj} $ are independent and identical for all $ j $, and the same is true for the distributions of $ \bU_{j} $. Define $ \bar{\bU} = n^{-1} \sum_{j=1}^{n} \bU_j $ and $ \bar{\bU}_q = n^{-1} \sum_{j=1}^{n} \bU_{qj} $. It follows from the multivariate central limit theorem that
$ \sqrt{n} \bar{\bU} \stackrel{d}{\rightarrow} \bW $ and $ \sqrt{n} \bar{\bU}_{q_0} \stackrel{d}{\rightarrow} \bW_{q_0} $ as $ n \to \infty $.
Now, note that $ \max_{q, i} n \tilde{b}_{1,q,i} = G( \sqrt{n} \bar{\bU} ) $ and $ \max_{i} n \tilde{b}_{1, q_0, i} = \sqrt{n} \bar{\bU}_{q_0} $. Since $ G( \cdot ) $ and $ G_q( \cdot ) $ are continuous, the proof of the theorem is completed from an application of the continuous mapping theorem.
\end{proof}

We implement the tests of skewness under Gaussianity of the null hypotheses.
To derive the p-values of the tests of the hypotheses described in \eqref{h_skewness} and \eqref{h_skewness_q0}, we need to first estimate $ \bOmega $ and $ \bOmega_{q_0} $, the dispersion matrices of $ \bU_j $ and $ \bU_{q_0 j} $, respectively.
However, the random vectors $ \bU_j $ and $ \bU_{q_0 j} $ are constituted of $\bu_{qij}$, whose definition involve unknown population quantities $ \sigma_i $, $ \mu_i $, $ \lambda_{qik} $ and $ f_{qik}( \bX_{qij} ) $.
So, we substitute $\bu_{qij}$ in $ \bU_j $ and $ \bU_{q_0 j} $ by $\widehat{\bu}_{qij}$ for all $q, i \text{ and } j$ to form $ \widehat{\bU}_j $ and $ \widehat{\bU}_{q_0 j} $, respectively, where $\widehat{\bu}_{qij}$ is constructed based on the sample observations. The construction procedure of $\widehat{\bu}_{qij}$ is described below.

Let $\bar{X}_i$ and $s_i^2$ denote the sample mean and sample variance, respectively, of the $i$-th components $X_{ij}$ of the observation vectors $\bX_j$, where $j \in \{1, \ldots, n\}$ and $i \in \{1, \ldots, p\}$.
When $ q = 1 $, $\bu_{qij}$ is a univariate random variable, and it can be verified that under Gaussianity, $ \text{Var}( \bu_{qij} ) = 6 $ for all $ i \text{ and } j $.
For $ q = 1 $, we construct $\widehat{\bu}_{qij}$ ensuring that its sample variance is also equal to $6$. We first compute $ \tilde{u}_{ij} = s_i^{-3} ( X_{ij} - \bar{X}_i) \{( X_{ij} - \bar{X}_i)^2 - 3 s_i^{2}\} $ for all $j$, and then set $ \widehat{\bu}_{qij} = \left\{ \sqrt{6} / \tilde{s}_i \right\} \tilde{u}_{ij} $, where $ \tilde{s}_i^2 $ denotes the sample variance of $ \tilde{u}_{ij} $. In this way, it is ensured that the sample variance of $\widehat{\bu}_{qij}$ is equal to $ \text{Var}( \bu_{qij} ) $ under Gaussianity for all $i \text{ and } j$ when $q = 1$.
Next, for $q > 1$, $ \lambda_{qik} $ is derived under Gaussianity, and it can be verified that $ \lambda_{qik} = 6 $ for all $ q, i, k $ in this case.
Recall that the function $ f_{qik}( \cdot ) $ is the eigenfunction of the integral operator $ \mathbf{h}_{qi} $ corresponding to the eigenvalue $ \lambda_{qik} $. From the arguments in Theorem 2.2 in \cite{baringhaus1992limit}, it follows that $ f_{qik}( \cdot ) $ are spherical harmonic functions (see \cite{erdelyi1953higher}). Explicit expressions of these spherical harmonic functions can be derived (see \cite{higuchi1987symmetric}, \cite{blanco1997evaluation}). However, their numerical approximation may be unstable due to the involvement of hypergeometric functions (see \cite[p.~1554]{higuchi1987symmetric}). For this reason and ease of computation, the random variables $ f_{qik}( \bX_{qij} ) $ in the definition of $\bu_{qij}$ are substituted in the following way. First, we form the $ n \times n $ matrix $ \widehat{{\bf H}}_{qi} = ( \widehat{h}_{qi}( \bX_{qij}, \bX_{qil} )) $, where:
\begin{align*}
\widehat{h}_{qi}( \bX_{qij}, \bX_{qil} ) 
&= \left\{ ( \bX_{qij} - \bar{\bX}_{qi} )^\top \mathbf{S}_{qi}^{-1} ( \bX_{qil} - \bar{\bX}_{qi} ) \right\}^3+ 3 (q + 2) \left\{ ( \bX_{qij} - \bar{\bX}_{qi} )^\top \mathbf{S}_{qi}^{-1} ( \bX_{qil} - \bar{\bX}_{qi} ) \right\} \nonumber \\
& \quad - 3 ( \bX_{qij} - \bar{\bX}_{qi} )^\top \mathbf{S}_{qi}^{-1} ( \bX_{qij} - \bar{\bX}_{qi} ) ( \bX_{qij} - \bar{\bX}_{qi} )^\top \mathbf{S}_{qi}^{-1} ( \bX_{qil} - \bar{\bX}_{qi} ) \nonumber\\
& \quad - 3 ( \bX_{qil} - \bar{\bX}_{qi} )^\top \mathbf{S}_{qi}^{-1} ( \bX_{qil} - \bar{\bX}_{qi} ) ( \bX_{qij} - \bar{\bX}_{qi} )^\top \mathbf{S}_{qi}^{-1} ( \bX_{qil} - \bar{\bX}_{qi} ). \nonumber
\end{align*}
We then compute the $ K(q) $ eigenvectors of the matrix $ \widehat{{\bf H}}_{qi} $ corresponding to the $ K(q) $ eigenvalues with the largest magnitudes, and arrange them by the descending order of the magnitudes of their corresponding eigenvalues. Each computed eigenvector is multiplied by $ \sqrt{n} $ to maintain its correspondence with the eigenfunctions of the integral operator $ \mathbf{h}_{qi} $.
Each such vector obtained, denoted as $ \widehat{f}_{qik} $, substitutes the vector $ ( f_{qik}( \bX_{qi1} ), \ldots, f_{qik}( \bX_{qin} ) )^\top $.
From the components of the vectors $ \widehat{f}_{qik} $, we compute $\widehat{\bu}_{qij}$ analogous to how $\bu_{qij}$ is defined using the components of the vectors $ ( f_{qik}( \bX_{qi1} ), \ldots, f_{qik}( \bX_{qin} ) )^\top $ and take $ \lambda_{qik} = 6 $. Then, we construct the vectors $ \widehat{\bU}_j $ and $ \widehat{\bU}_{q_0 j} $ using $\widehat{\bu}_{qij}$ in place of $\bu_{qij}$ in $ \bU_j $ and $ \bU_{q_0 j} $, respectively. Finally, the dispersion matrices $ \bOmega $ and $ \bOmega_{q_0} $ are estimated by:
\begin{align*}
 \widehat{\bOmega} &= (n - 1)^{-1} \sum_{j=1}^{n} \left( \widehat{\bU}_j - n^{-1} \sum_{l=1}^{n} \widehat{\bU}_l \right) \left( \widehat{\bU}_j - n^{-1} \sum_{l=1}^{n} \widehat{\bU}_l \right)^\top , \\
 \widehat{\bOmega}_{q_0} &= (n - 1)^{-1} \sum_{j=1}^{n} \left( \widehat{\bU}_{q_0 j} - n^{-1} \sum_{l=1}^{n} \widehat{\bU}_{q_0 l} \right) \left( \widehat{\bU}_{q_0 j} - n^{-1} \sum_{l=1}^{n} \widehat{\bU}_{q_0 l} \right)^\top .
\end{align*}

To compute the p-values of the MaxS test for \eqref{h_skewness}, we generate 1000 independent zero-mean Gaussian random vectors $ \tilde{\bW}_1, \ldots, \tilde{\bW}_{1000} $ with the dispersion matrix $ \widehat{\bOmega} $. The proportion of the values $ G( \tilde{\bW}_1 ), \ldots, G( \tilde{\bW}_{1000} ) $ larger than $ \max_{q, i} n \tilde{b}_{1,q,i} $ is taken as the p-value of the null hypothesis in \eqref{h_skewness}.
The p-value of the test $ \text{MaxS}_{q_0} $ in \eqref{h_skewness_q0} is derived similarly.

\subsection{Kurtosis}

Next, we derive the asymptotic null distributions of $ \max_{q, i} | \tilde{b}_{2,q,i}| $ and of $ \max_{i}| \tilde{b}_{2, q_0, i}| $. Here also, we first derive a linearization of $ b_{2,q,i} $.

\begin{theorem}\label{thm:kurtosis1}
Let $ \mathbf{X}_1, \ldots, \mathbf{X}_n $ be independent and identically distributed random vectors in $ \mathbb{R}^p $ with $ \text{E}( \| \mathbf{X} \|^8 ) < \infty $.
Let $ \bA_{qi} = \bSigma_{qi}^{-1} \text{E}\{( \bX_{qi} - \bmu_{qi})( \bX_{qi} - \bmu_{qi})^\top \bSigma_{qi}^{-1}( \bX_{qi} - \bmu_{qi})( \bX_{qi} - \bmu_{qi})^\top \} $.
Define:
\begin{align*}
\bZ_{qij} =
\begin{bmatrix}
\left\{ \left( \bX_{qij} - \bmu_{qi} \right)^\top \bSigma_{qi}^{-1} \left( \bX_{qij} - \bmu_{qi} \right) \right\}^2 - \text{E}\left[ \left\{ \left( \bX_{qi} - \bmu_{qi} \right)^\top \bSigma_{qi}^{-1} \left( \bX_{qi} - \bmu_{qi} \right) \right\}^2 \right]  \\
\left( \bX_{qij} - \bmu_{qi} \right)^\top \bA_{qi} \bSigma_{qi}^{-1} \left( \bX_{qij} - \bmu_{qi} \right) - \text{E}\left\{ \left( \bX_{qij} - \bmu_{qi} \right)^\top \bA_{qi} \bSigma_{qi}^{-1} \left( \bX_{qij} - \bmu_{qi} \right) \right\} \\
\bX_{qij} - \bmu_{qi}
\end{bmatrix},
\end{align*}
and
\begin{align*}
\ba_{qi} = \left( 1,-2,-4 \text{E}\left[ \left\{ \left( \bX_{qi} - \bmu_{qi} \right)^\top \bSigma_{qi}^{-1} \left( \bX_{qi} - \bmu_{qi} \right) \right\} \bSigma_{qi}^{-1} \left( \bX_{qi} - \bmu_{qi} \right) \right] \right)^{\top} .
\end{align*}
Then,
$ n^{1/2}(b_{2,q,i}-\beta_{2,q,i}) = n^{-1/2}\sum_{j=1}^n \ba_{qi}^{\top} \bZ_{qij} + o_P(1) $
as $ n \to \infty $.
\end{theorem}
\begin{proof}[\textbf{\upshape Proof:}]
Recall that $\beta_{2,q,i}$ and $b_{2,q,i}$ are invariant under location transformations. So, without loss of generality, we can assume $\bmu = \text{E}(\bX)=\0$, which means $\bmu_{qi}=\0$ and $ \text{Var}( \bX )=\text{E}( \bX\bX^\top ) = \bSigma$.

The arguments are similar to those in the proof of Theorem 2.1 in \cite{henze1994mardia}.
In this proof, $\bY_n=O_P(b_n)$ means that the sequence $\bY_n/b_n$ is bounded in Euclidean/matrix norm, while $\bY_n=o_P(b_n)$ means that $b_n^{-1}\bY_n \stackrel{\text{P}}{\to}0$.
Let:
\begin{align*}
\bB_{q i} = n^{1/2} \left\{ n^{-1} \sum_{j=1}^{n} \left( \bX_{qij} - \bmu_{qi} \right) \left( \bX_{qij} - \bmu_{qi} \right)^\top - \bSigma_{qi} \right\} .
\end{align*}
From the weak law of large numbers, we get
$ \bB_{q i} = O_P( 1 ) $ for all $ q, i $.
Also:
\begin{align}
n^{1/2} \left( \bS_{qi} - \bSigma_{qi} \right) = \bB_{q i} - n^{1/2} \left( \bar{\bX}_{qi} - \bmu_{qi} \right) \left( \bar{\bX}_{qi} - \bmu_{qi} \right)^\top .
\label{kurteq1}
\end{align}
From the multivariate central limit theorem, we have that $ n^{1/2} \left( \bar{\bX}_{qi} - \bmu_{qi} \right) \left( \bar{\bX}_{qi} - \bmu_{qi} \right)^\top = O_P( n^{-1/2} ) $, which, applied to \eqref{kurteq1}, yields:
\begin{align*}
& \bS_{qi} = \bSigma_{qi} + n^{-1/2} \bB_{q i} + O_P( n^{-1} ) , \text{or, }
\bSigma_{qi}^{-1} \bS_{qi} = \bI_{q} + n^{-1/2} \bSigma_{qi}^{-1} \bB_{q i} + O_P( n^{-1} ) .
\end{align*}

Note that
 $n^{1/2}(\bS_{qi} - \bSigma_{qi}) =\bB_{qi} - n^{1/2} \bar{\bX}_{qi} \bar{\bX}_{qi}^{\top}$
and $n^{1/2} \bar{\bX}_{qi} \bar{\bX}_{qi}^{\top} = O_P(n^{-1/2})$ by the multivariate central limit theorem. It follows that:
$$ \bS_{qi} = \bSigma_{qi} + n^{-1/2} \bB_{qi} + O_P(n^{-1}),$$
and thus:
$$\bSigma_{qi}^{-1} \bS_{qi} = \bI_{qi} + n^{-1/2} \bSigma_{qi}^{-1}\bB_{qi} + O_P(n^{-1}).$$
This means $ \bSigma_{qi}^{-1} \bS_{qi} $ is invertible for all sufficiently large $ n $ with probability approaching 1, and we have:
$$(\bSigma_{qi}^{-1} \bS_{qi})^{-1} = \bI_{qi} - n^{-1/2} \bSigma_{qi}^{-1}\bB_{qi} + O_P(n^{-1}),$$
which implies that:
\begin{equation} \label{EQ1}
\bS_{qi}^{-1} = \bSigma_{qi}^{-1} - n^{-1/2} \bSigma_{qi}^{-1}\bB_{qi}\bSigma_{qi}^{-1} + O_P(n^{-1}).
\end{equation}
Now,
\begin{align}
 \{(\bX_{qij}-\bar{\bX}_{qi})^\top \bS_{qi}^{-1}(\bX_{qij}-\bar{\bX}_{qi})\}^2
&= (\bX_{qij}^\top \bS_{qi}^{-1}\bX_{qij})^2 +4 (\bar{\bX}_{qi}^\top \bS_{qi}^{-1}\bX_{qij})^2 + (\bar{\bX}_{qi}^\top \bS_{qi}^{-1} \bar{\bX}_{qi})^2 \nonumber \\
& \quad - 4 (\bX_{qij}^\top \bS_{qi}^{-1}\bX_{qij} \bar{\bX}_{qi}^\top \bS_{qi}^{-1} \bX_{qij}) + 2 (\bX_{qij}^\top \bS_{qi}^{-1} \bX_{qij} \bar{\bX}_{qi}^\top \bS_{qi}^{-1} \bar{\bX}_{qi}) \nonumber \\
& \quad -  4 (\bar{\bX}_{qi}^\top \bS_{qi}^{-1} \bX_{qij} \bar{\bX}_{qi}^\top \bS_{qi}^{-1} \bar{\bX}_{qi}). \label{EQ2}
\end{align}

Using this and the expression for the inverse in \eqref{EQ1} above, we get for the first term in \eqref{EQ2} to be:
$$ \bX_{qij}^\top \bS_{qi}^{-1}\bX_{qij} = \bX_{qij}^\top \bSigma_{qi}^{-1}\bX_{qij} - n^{-1/2} \bX_{qij}^\top \bSigma_{qi}^{-1}\bB_{qi}\bSigma_{qi}^{-1}\bX_{qij} + O_P(n^{-1}).$$
Squaring both sides, we get:
\begin{align*}
(\bX_{qij}^\top \bS_{qi}^{-1}\bX_{qij})^2 
&= (\bX_{qij}^\top \bSigma_{qi}^{-1}\bX_{qij})^2 - 2 n^{-1/2} \bX_{qij}^\top \bSigma_{qi}^{-1}\bX_{qij} \bX_{qij}^\top \bSigma_{qi}^{-1}\bB_{qi}\bSigma_{qi}^{-1} \bX_{qij} +  O_P(n^{-1}).
\end{align*}
Using the fact that the trace of a matrix is invariant under cyclic permutations, we have:
\begin{align*}
\text{tr}(\bX_{qij}^\top \bSigma_{qi}^{-1}\bX_{qij} \bX_{qij}^\top \bSigma_{qi}^{-1}\bB_{qi}\bSigma_{qi}^{-1}\bX_{qij}) = \text{tr}\{\bSigma_{qi}^{-1}\bB_{qi} (\bSigma_{qi}^{-1} \bX_{qij} \bX_{qij}^\top) (\bSigma_{qi}^{-1} \bX_{qij} \bX_{qij}^\top)\}.
\end{align*}
Since $n^{-1} \sum_{j=1}^n ( \bSigma_{qi}^{-1}\bX_{qij} \bX_{qij}^\top) ( \bSigma_{qi}^{-1}\bX_{qij} \bX_{qij}^\top) = \bA_{qi} + o_P(1)$ with $ \bA_{qi}= \text{E}\{ (\bSigma_{qi}^{-1}\bX_{qi} \bX_{qi}^\top) (\bSigma_{qi}^{-1}\bX_{qi} \bX_{qi}^\top) \} $, we now get:
\begin{align*}
\frac{1}{n} \sum_{j=1}^n (\bX_{qij}^\top \bS_{qi}^{-1}\bX_{qij})^2 = \frac{1}{n} \sum_{j=1}^n (\bX_{qij}^\top \bSigma_{qi}^{-1}\bX_{qij})^2 - 2 n^{-1/2} \text{tr}(\bSigma_{qi}^{-1}\bB_{qi} \bA_{qi}) + o_P(n^{-1}) .
\end{align*}

For the second term in \eqref{EQ2}, again using equation \eqref{EQ1} and the fact that $n^{-1} \sum_{j=1}^n \bX_{qij} \bX_{qij}^\top = O_P(1)$ (using the weak law of large numbers) we obtain:
\begin{align*}
\frac{1}{n} \sum_{j=1}^n (\bar{\bX}_{qi}^\top \bS_{qi}^{-1} \bX_{qij})^2 = O_P(n^{-1/2}).
\end{align*}
Furthermore, for the third term in \eqref{EQ2}:
$
(\bar{\bX}_{qi}^\top \bS_{qi}^{-1} \bar{\bX}_{qi})^2 = O_P(n^{-2})
$.
Since, $n^{-1} \sum_{j=1}^n \bSigma_{qi}^{-1} \bX_{qij} (\bX_{qij}^\top \bSigma_{qi}^{-1} \bX_{qij}) = \tilde{\ba}_{qi} + o_P(1)$ with $ \tilde{\ba}_{qi} = \text{E}\{ \bSigma_{qi}^{-1} \bX_{qi} (\bX_{qi}^\top \bSigma_{qi}^{-1} \bX_{qi}) \} $, it is easy to see that for the fourth term:
\begin{align*}
\frac{1}{n} \sum_{j=1}^n (\bX_{qij}^\top \bS_{qi}^{-1} \bX_{qij}) (\bar{\bX}_{qi}^\top \bS_{qi}^{-1} \bX_{qij}) = \tilde{\ba}_{qi}^\top \bar{\bX}_{qi} + O_P(n^{-1}).
\end{align*}
Finally, for the last two terms:
\begin{align*}
& \frac{1}{n} \sum_{j=1}^n (\bX_{qij}^\top \bS_{qi}^{-1} \bX_{qij})(\bar{\bX}_{qi}^\top \bS_{qi}^{-1} \bar{\bX}_{qi}) = O_P(n^{-1}) 
\;\text{and}\;
 \frac{1}{n} \sum_{j=1}^n (\bar{\bX}_{qi}^\top \bS_{qi}^{-1} \bX_{qij})(\bar{\bX}_{qi}^\top \bS_{qi}^{-1} \bar{\bX}_{qi}) = O_P(n^{-1}).
\end{align*}

Summarizing, we obtain:
\begin{align*}
b_{2,q,i} =  \frac{1}{n} \sum_{j=1}^n (\bX_{qij}^\top \bSigma_{qi}^{-1}\bX_{qij})^2 - 2 n^{-1/2} \text{tr}(\bSigma_{qi}^{-1}\bB_{qi} \bA_{qi}) - 4\tilde{\ba}_{qi}^\top \bar{\bX}_{qi} + O_P(n^{-1}).
\end{align*}
Observing that
\begin{align*}
\text{tr}(\bSigma_{qi}^{-1}\bB_{qi} \bA_{qi}) 
&= \text{tr}(\bB_{qi} \bA_{qi}\bSigma_{qi}^{-1}) 
= n^{1/2} \left\{ \frac{1}{n} \sum_{j=1}^n \bX_{qij}^\top \bA_{qi}\bSigma_{qi}^{-1} \bX_{qij} - \text{E}( \bX_{qi}^\top \bA_{qi}\bSigma_{qi}^{-1} \bX_{qi} ) \right\}
\end{align*}
we obtain
\begin{align*}
n^{1/2}(b_{2,q,i}-\beta_{2,q,i}) = n^{-1/2}\sum_{j=1}^n  \ba_{qi}^{\top} \bZ_{qij} + o_P(1),
\end{align*}
where
\begin{align*}
\bZ_{qij} =
\begin{bmatrix}
\left(\bX_{qij}^\top \bSigma_{qi}^{-1} \bX_{qij}\right)^2 - \text{E}\{( \bX_{qi}^\top \bSigma_{qi}^{-1} \bX_{qi} )^2\}  \\
\bX_{qij}^\top \bA_{qi}\bSigma_{qi}^{-1} \bX_{qij} - \text{E}( \bX_{qi}^\top \bA_{qi} \bSigma_{qi}^{-1} \bX_{qi} ) \\
\bX_{qij}
\end{bmatrix}
\end{align*}
is a $(2+q)$-dimensional vector and $\ba_{qi}=(1,-2,-4\tilde{\ba}_{qi}^\top)^{\top}$.
\end{proof}

Tests of kurtosis are usually conducted for testing normality (see \cite{henze1994mardia}). When the underlying distribution is Gaussian, we can simplify the quantities $ \bA_{qi} $ and $ \ba_{qi} $ described in \autoref{thm:kurtosis1}, and consequently, the linearization becomes simpler.
\begin{corollary}\label{coro:kurtosis2}
Let $ \mathbf{X}_1, \ldots, \mathbf{X}_n $ be independent and identically distributed Gaussian random vectors in $ \mathbb{R}^p $.
Define
$
Y_{qij} = \left\{( \bX_{qij} - \bmu_{qi})^\top \bSigma_{qi}^{-1} ( \bX_{qij} - \bmu_{qi} ) \right\}^2 - 2 (q + 2) ( \bX_{qij} - \bmu_{qi} )^\top \bSigma_{qi}^{-1} ( \bX_{qij} - \bmu_{qi}) .
$
Then,
$ n^{1/2}(b_{2,q,i}-\beta_{2,q,i}) = n^{-1/2} \sum_{j=1}^n \{ Y_{qij} - \text{E}( Y_{qij} ) \} + o_P(1) $ as $ n \to \infty $.
\end{corollary}
\begin{proof}[\textbf{\upshape Proof:}]
Define $ \tilde{\bX}_{qij} = \bSigma_{qi}^{-1/2} ( \bX_{qij} - \bmu_{qi}) $. Under the assumption of the corollary, the $ \tilde{\bX}_{qij} $s are independent zero-mean Gaussian random vectors with the identity matrix as their dispersion matrix. It follows that:
\begin{align}
 \text{E}\left[ \left\{ \left( \bX_{qi} - \bmu_{qi} \right)^\top \bSigma_{qi}^{-1} \left( \bX_{qi} - \bmu_{qi} \right) \right\} \bSigma_{qi}^{-1} \left( \bX_{qi} - \bmu_{qi} \right) \right] 
&= \bSigma_{qi}^{-1/2} \text{E}\left\{ \left( \tilde{\bX}_{qi1}^\top \tilde{\bX}_{qi1} \right) \tilde{\bX}_{qi1} \right\} \nonumber \\
&= \bSigma_{qi}^{-1/2} \times \bf0 = \bf0 .
\label{eq:coro1}
\end{align}
Next, we have:
\begin{align}
& \left( \bX_{qij} - \bmu_{qi} \right)^\top \bA_{qi} \bSigma_{qi}^{-1} \left( \bX_{qij} - \bmu_{qi} \right) 
= \tilde{\bX}_{qij}^\top \text{E}\left[ \tilde{\bX}_{qij} \left( \tilde{\bX}_{qij}^\top \tilde{\bX}_{qij} \right) \tilde{\bX}_{qij}^\top \right] \tilde{\bX}_{qij}
= (q + 2) \tilde{\bX}_{qij}^\top \tilde{\bX}_{qij} .
\label{eq:coro2}
\end{align}
The proof follows from \eqref{eq:coro1}, \eqref{eq:coro2} and the linearization in \autoref{thm:kurtosis1}.
\end{proof}

The asymptotic null distributions of $ \max_{q, i} | \tilde{b}_{2,q,i}| $ and $ \max_{i}| \tilde{b}_{2, q_0, i} | $ are derived from the linearization in \autoref{coro:kurtosis2}.
\begin{theorem}
Let $ \mathbf{X}_1, \ldots, \mathbf{X}_n $ be independent and identically distributed Gaussian random vectors in $ \mathbb{R}^p $.
Define:
\begin{align*}
& \tilde{\bY}_j = \left( \frac{Y_{1 1 j}}{\sqrt{24}}, \ldots, \frac{Y_{q i j}}{\sqrt{8 q (q + 2)}}, \ldots, \frac{Y_{p 1 j}}{\sqrt{8 p (p + 2)}} \right)^\top ,\\
& \tilde{\bY}_{q j} = \left( \frac{Y_{q 1 j}}{\sqrt{8 q (q + 2)}}, \ldots, \frac{Y_{q Q_q j}}{\sqrt{8 q (q + 2)}} \right)^\top ,
\end{align*}
where $Y_{q i j}$ is as defined in \autoref{coro:kurtosis2}.
Let $ \bGamma $ and $ \bGamma_q $ be the dispersion matrices of $ \tilde{\bY}_1 $ and $ \tilde{\bY}_{q1} $, respectively. Let $ \bW $ and $ \bW_q $ be zero-mean Gaussian random vectors with dispersion matrices $ \bGamma $ and $ \bGamma_q $, respectively.
Then,
$ \sqrt{n} \max_{q, i} | \tilde{b}_{2,q,i} | \stackrel{d}{\rightarrow} \| \bW \|_\infty $ and $ \sqrt{n} \max_{i} | \tilde{b}_{2, q_0, i}| \stackrel{d}{\rightarrow} \| \bW_q \|_\infty $ as $ n \to \infty $, respectively, where $\| \cdot \|_\infty$ is the $ l_\infty $ norm in the Euclidean space, i.e., $ \| \bv \|_\infty = \max_i | v_i | $ with $v_i$ being the $i$-th component of the vector $\bv$.
\end{theorem}
\begin{proof}[\textbf{\upshape Proof:}]
Under the assumption of Gaussianity in the theorem, $ \beta_{2,q,i} = q (q + 2) $ for all $ q $. Since the $ l_\infty $ norm in the Euclidean space is continuous, the proof follows from an application of the multivariate central limit theorem in \autoref{coro:kurtosis2} and then applying the continuous mapping theorem.
\end{proof}

The tests of kurtosis are implemented under Gaussianity of the null hypotheses.
To compute the p-values for the MaxK test in \eqref{h_kurtosis} and the $ \text{MaxK}_{q_0} $ test in \eqref{h_kurtosis_q0}, we need to estimate the dispersion matrices $ \bGamma $ and $ \bGamma_q $ of the random vectors $ \tilde{\bY}_j $ and $ \tilde{\bY}_{q j} $. But, $ \tilde{\bY}_j $ and $ \tilde{\bY}_{q j} $ are constituted of $Y_{qij}$, whose definition involves unknown population quantities $\bmu_{qi}$ and $\bSigma_{qi}$. So, we substitute $ Y_{qij} $ in $ \tilde{\bY}_j $ and $ \tilde{\bY}_{q j} $ by:
\begin{align*}
\widehat{Y}_{qij} = \{ ( \bX_{qij} - \bar{\bX}_{qi} )^\top \mathbf{S}_{qi}^{-1} ( \bX_{qij} - \bar{\bX}_{qi} )\}^2 - 2 (q + 2) ( \bX_{qij} - \bar{\bX}_{qi})^\top \mathbf{S}_{qi}^{-1} ( \bX_{qij} - \bar{\bX}_{qi}) ,
\end{align*}
which is obtained by replacing $\bmu_{qi}$ and $\bSigma_{qi}$ in the expression of $ Y_{qij} $ by their estimates $\bar{\bX}_{qi}$ and $ \mathbf{S}_{qi} $, respectively. Now, it can be derived that under Gaussianity, $ \text{Var}( Y_{qi1} ) = 8 q (q + 2) $ for all $ q, i $. So, the diagonal entries of the estimates of $ \bGamma $ and $ \bGamma_q $, denoted as $ \widehat{\bGamma} $ and $ \widehat{\bGamma}_q $, respectively, are fixed to be 1. The off-diagonal entries of $ \widehat{\bGamma} $ and $ \widehat{\bGamma}_q $ are obtained from the sample correlations of $\widehat{Y}_{q_1 i_1 j}$ and $\widehat{Y}_{q_2 i_2 j}$ for the corresponding indices $q_1, i_1$ and $q_2, i_2$.
Next, we generate 1000 independent zero mean Gaussian random vectors $ \tilde{\bW}_1, \ldots, \tilde{\bW}_{1000} $ with dispersion matrix $ \widehat{\bGamma} $. The p-value of the MaxK test in \eqref{h_kurtosis} is the proportion of the values $ \| \tilde{\bW}_1 \|_\infty, \ldots, \| \tilde{\bW}_{1000} \|_\infty $ larger than $ \max_{q, i}| \tilde{b}_{2,q,i}| $. The p-value of the $ \text{MaxK}_{q_0} $ test in \eqref{h_kurtosis_q0} is computed similarly.

\subsection{Testing Gaussianity Based on Both Skewness and Kurtosis}\label{subsec:gaussianitytest}

Based on the tests of skewness and kurtosis, a test of Gaussianity can be constructed. Analogous to \eqref{h_skewness} and \eqref{h_kurtosis}, the null hypothesis here is:
\begin{align}
\text{H}_0^{(g)} : \text{The underlying distribution is Gaussian}.
\label{h_gaussianity}
\end{align}
Similarly, analogous to \eqref{h_skewness_q0} and \eqref{h_kurtosis_q0}, the null hypothesis is:
\begin{align}
\text{H}_0^{q_0,(g)} : \text{All $ q_0 $-dimensional subsets of the data follow some Gaussian distribution}.
\label{h_gaussianity_q0}
\end{align}
Here, if at least one of the null hypotheses of the corresponding skewness test or kurtosis test is rejected, then we reject the null hypothesis of Gaussianity (after Bonferroni correction). For example, if any of the MaxS test and the MaxK test rejects their null hypotheses, then \eqref{h_gaussianity} is also rejected. The test for \eqref{h_gaussianity} is denoted as MaxSK test. Similarly, if any of the $ \text{MaxS}_{q_0} $ test or the $ \text{MaxK}_{q_0} $ test rejects their null hypothesis for a fixed $ q_0 $, then \eqref{h_gaussianity_q0} is rejected, and we denote this test by $ \text{MaxSK}_{q_0} $ test.

\section{Simulation Study}\label{sec:simulation}

In this section, the performance of our proposed tests are investigated in terms of the estimated sizes and powers using some simulated models. The estimated powers of our tests are compared with the corresponding Mardia tests. We also compare the estimated powers of ours tests with several tests of Gaussianity.

\subsection{Estimated Sizes}\label{subsec:simulation-1}
For estimating the sizes of our tests, we consider $ \mathbf{X}_1, \ldots, \mathbf{X}_n $ being a random sample from the $ p $-variate Gaussian distribution $ {\cal N}_p( \mathbf{0}, \boldsymbol{\Sigma} ) $, where $ \boldsymbol{\Sigma} = ( \sigma_{ij} ) $ with $ \sigma_{ij} = 0.5 + 0.5 \mathbb{I}( i = j ) $. We take $ p = 5 $. The number of replicates to estimate the sizes of our tests is taken as $ 1000 $. The sample size $ n $ is varied.

\begin{table}[b!]
\begin{center}
\caption{Estimated sizes of the tests for 5\% nominal level in $ {\cal N}_5( \mathbf{0}, \boldsymbol{\Sigma} ) $ based on 1000 replicates, where $ \boldsymbol{\Sigma} = ( \sigma_{ij} ) $ with $ \sigma_{ij} = 0.5 + 0.5 \mathbb{I}( i = j ) $.}
\label{tab:size}
\begin{tabular} {l c c c c c}  
\hline 
Test & $ n = 50 $ & $ n = 100 $ & $ n = 200 $ & $ n = 500 $ & $ n = 1000 $ \\ \hline 
MaxS & 0.048 & 0.056 & 0.048 & 0.047 & 0.044 \\ 
MaxK & 0.043 & 0.049 & 0.058 & 0.060 & 0.049 \\ 
MaxSK & 0.044 & 0.053 & 0.057 & 0.055 & 0.045 \\ 
$ \text{MaxS}_1 $ & 0.041 & 0.053 & 0.043 & 0.051 & 0.046 \\ 
$ \text{MaxS}_2 $ & 0.050 & 0.062 & 0.051 & 0.055 & 0.040 \\ 
$ \text{MaxS}_3 $ & 0.034 & 0.059 & 0.055 & 0.055 & 0.035 \\ 
$ \text{MaxS}_4 $ & 0.030 & 0.050 & 0.043 & 0.055 & 0.038 \\ 
$ \text{MaxS}_5 $ & 0.015 & 0.041 & 0.043 & 0.045 & 0.048 \\ 
$ \text{MaxK}_1 $ & 0.047 & 0.062 & 0.053 & 0.058 & 0.048 \\ 
$ \text{MaxK}_2 $ & 0.021 & 0.034 & 0.045 & 0.057 & 0.050 \\ 
$ \text{MaxK}_3 $ & 0.014 & 0.018 & 0.030 & 0.049 & 0.051 \\ 
$ \text{MaxK}_4 $ & 0.031 & 0.038 & 0.036 & 0.059 & 0.045 \\ 
$ \text{MaxK}_5 $ & 0.151 & 0.088 & 0.060 & 0.069 & 0.050 \\ 
$ \text{MaxSK}_1 $ & 0.048 & 0.065 & 0.053 & 0.049 & 0.048 \\ 
$ \text{MaxSK}_2 $ & 0.033 & 0.049 & 0.049 & 0.049 & 0.050 \\ 
$ \text{MaxSK}_3 $ & 0.025 & 0.039 & 0.039 & 0.052 & 0.044 \\ 
$ \text{MaxSK}_4 $ & 0.022 & 0.035 & 0.035 & 0.055 & 0.039 \\ 
$ \text{MaxSK}_5 $ & 0.070 & 0.058 & 0.046 & 0.057 & 0.041 \\ \hline 
\end{tabular}
\end{center}
\end{table}

From \autoref{tab:size}, it can be seen that the tests MaxS, MaxK and MaxSK have estimated sizes close to the 5\% nominal level for $ n = 50 $ and up. The estimated sizes of $ \text{MaxS}_5 $ and $ \text{MaxK}_2 $, $ \text{MaxK}_3 $, $ \text{MaxK}_5 $ and $ \text{MaxSK}_3 $, $ \text{MaxSK}_4 $ deviate slightly from the 5\% nominal level for $ n = 50 \text{ and } 100 $, and require higher sample sizes to converge to the nominal level.

In section 1 of the supplementary material, estimated sizes for some other values of the dimension $ p $ are presented, namely $ p = 3, 4 $ which correspond to the dimensions of the two datasets analyzed in \autoref{sec:realdata}. There also, it is found that the estimated sizes of the majority of the tests, including the tests MaxS, MaxK and MaxSK, are close to the 5\% nominal level for $ n = 50 $ and up. A few tests, like $ \text{MaxK}_2 $, $ \text{MaxK}_3 $, $ \text{MaxSK}_2 $, $ \text{MaxSK}_3 $, require higher sample sizes for their estimated sizes to reach close to the nominal level. In section 3 of the supplementary material, the estimated sizes of the tests are presented in an equicorrelation model, where the pairwise-correlation between the components of the vector $ \mathbf{X} $ is high. There, it is found that the estimated sizes of nearly all of the tests, including the tests MaxS, MaxK and MaxSK, are close to the nominal level for $ n = 100 $ and up. When the sample size is lower, i.e., $ n = 50 $, some deviations of the estimated sizes from the nominal level are observed.

\subsection{Estimated Powers}\label{subsec:simulation-2}

We now compare the powers of the tests with the usual Mardia skewness (MS) and kurtosis (MK) tests along with several tests of normality; see \cite{chen2022you} for a recent review.

The following tests of Gaussianity are considered for the comparison of performances.
The test by Henze and Zirkler \cite{henze1990class} is denoted as the HZ test. The test of normality developed by Royston \cite{royston1982extension,royston1983some,royston1992approximating,royston1995remark} is denoted as the R test. The testing procedure described by Doornik and Hansen \cite{doornik2008omnibus} is denoted as the DH test.
The skewness-based test of normality described by Kankainen, Taskinen and Oja \cite{kankainen2007tests} is denoted as the $ \text{KTO}_\text{S} $ test, and the kurtosis-based test of normality described in the same paper is denoted as the $ \text{KTO}_\text{K} $ test.
The test developed by Bowman and Shenton \cite{bowman1975omnibus} is denoted as the BS test.
The testing procedure studied by Villasenor Alva and Estrada \cite{villasenor2009generalization} is denoted as the VE test.
The test of normality developed by Zhou and Shao \cite{zhou2014powerful} is denoted as the ZS test.
The testing procedure based on the measure of skewness proposed in \cite{mori1993multivariate} is denoted as the MRS test.
The test for multivariate skewness described in \cite{malkovich1973tests} is denoted as the MAS test, while the test for multivariate kurtosis described by the same authors is denoted as the MAK test.
The test of kurtosis based on the measure of multivariate kurtosis proposed in \cite{koziol1989note} is denoted as the KK test.

The HZ test, the R test and the DH test are implemented using the corresponding functions in the \textsf{R} (\textsf{R} version 4.1.3 (2022-03-10), \cite{R}) package \texttt{MVN} \cite{MVN}. The $ \text{KTO}_\text{S} $ and $ \text{KTO}_\text{K} $ tests are implemented using the functions in the \textsf{R} package \texttt{ICS} \cite{ICS}. The BS test and the ZS test are implemented using their functions in the \textsf{R} package \texttt{mvnormalTest} \cite{mvnormalTest}. The VE test is implemented using its function in the \textsf{R} package \texttt{mvShapiroTest} \cite{mvShapiroTest}.
The MRS test is implemented using its function in the \textsf{R} package \texttt{MultiSkew} \cite{MultiSkew}.
The MAS test, the MAK test and the KK test are implemented using their respective functions in the \textsf{R} package \texttt{mnt} \cite{mnt}.
The Mardia skewness and kurtosis tests are implemented based on the asymptotic distributions of the test statistics derived in \cite{mardia1970measures} using the unbiased estimate of the population covariance matrix.

\begin{figure}[h!]
\centering
\includegraphics[height=5.9cm]{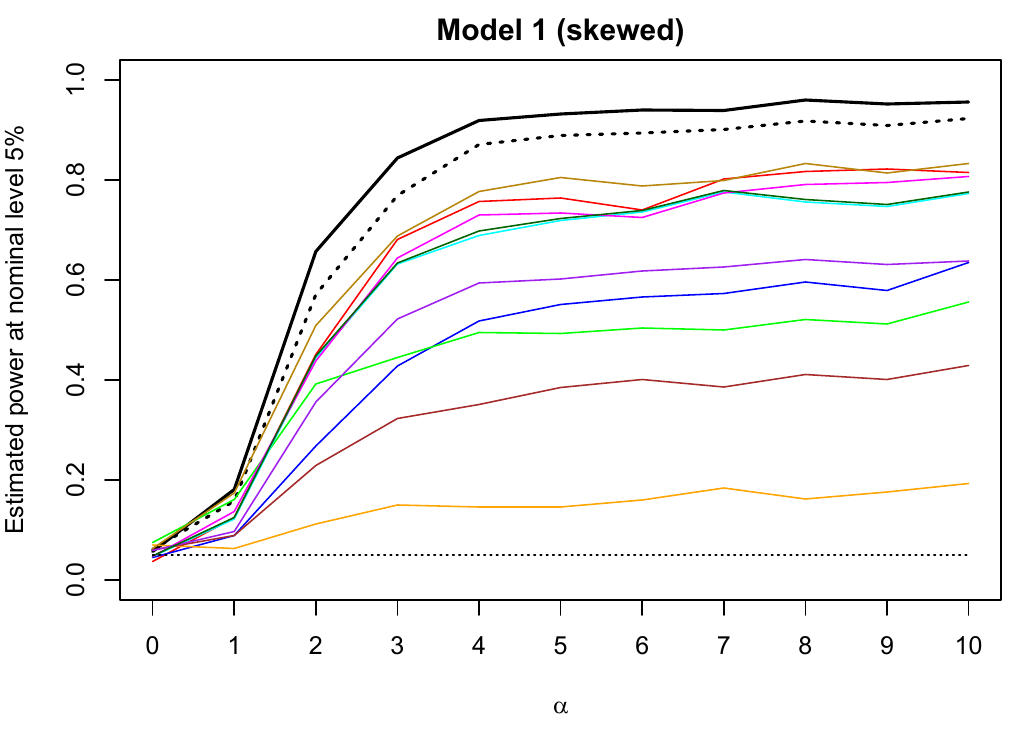}
\includegraphics[height=5.9cm]{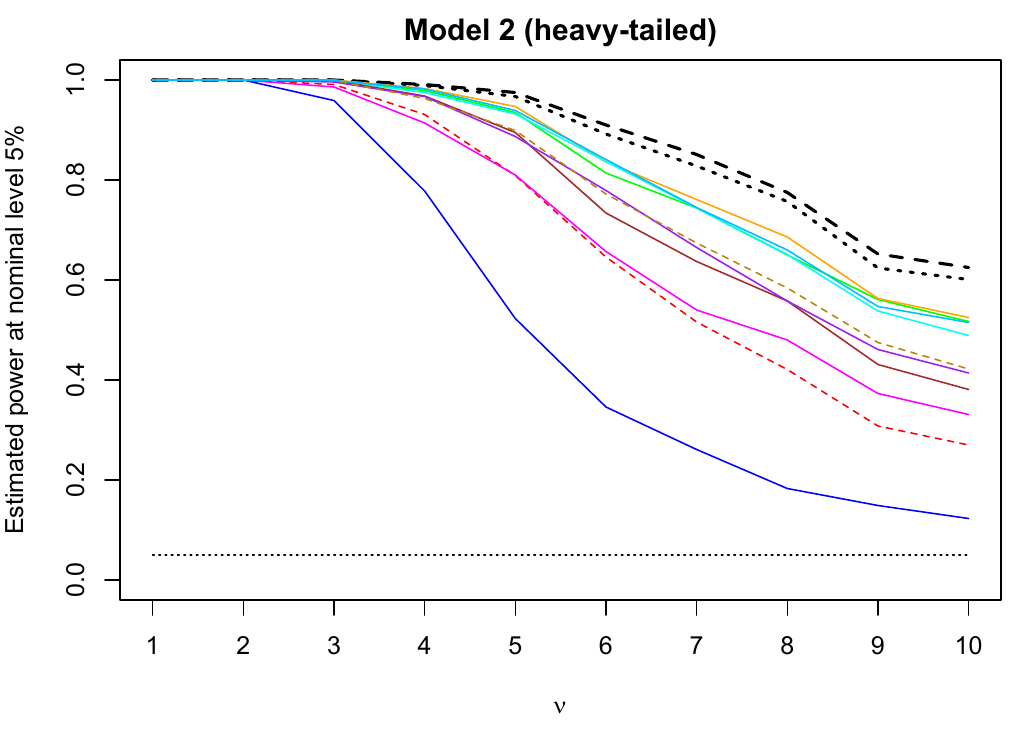}
\includegraphics[height=5.9cm]{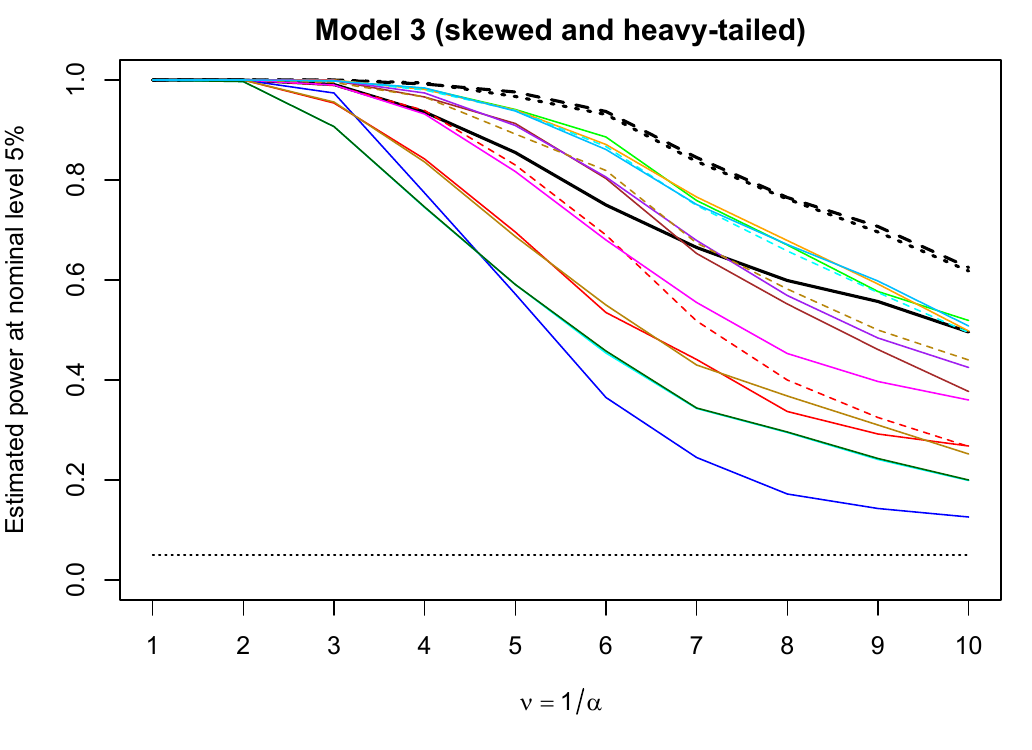}
\includegraphics[width=1\linewidth,trim={0 70 0 230},clip]{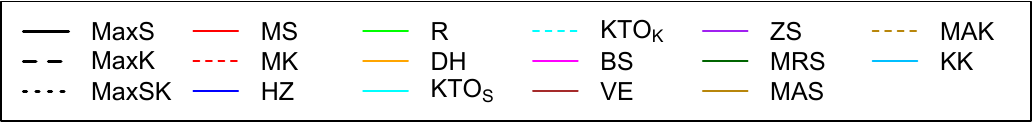}
\caption{Estimated powers of MaxS, MaxK, MaxSK, MS, MK, HZ, R, DH, $ \text{KTO}_\text{S} $, $ \text{KTO}_\text{K} $, BS, VE, ZS, MRS, MAS, MAK and KK tests for 5\% nominal level (horizontal dashed line) in Model 1 (top left), Model 2 (top right) and Model 3 (bottom) for $ n = 200 $, $p=5$, $q=2$ based on 1000 replicates. The horizontal dotted lines near the bottom of the plots correspond to the nominal level of 5\%.}
\label{fig:powerplots1}
\end{figure}

For the comparison of performances of the tests, we consider three simulation models. In each of the models, the non-Gaussian feature is supported on a small number of components of the random vector. Let $ \boldsymbol{\Sigma} = ( \sigma_{ij} ) $ with $ \sigma_{ij} = 0.5 + 0.5 \mathbb{I}( i = j ) $. The dimension of the matrix $ \boldsymbol{\Sigma} $ is to be determined based on the context.
Let $ \bX = (\bX_q^\top, \bX_{p-q}^\top)^\top $, where $ \bX_q $ and $ \bX_{p-q} $ are independent and $ \bX_{p-q} $ follows $ {\cal N}_{p-q}( \bf0, \boldsymbol{\Sigma} ) $. Then:
\begin{itemize}
\item Model 1 (skewed): $ \bX_q \sim {\cal SN}_q( \bf0, \boldsymbol{\Sigma}, \balpha ) $, where $ \balpha = \alpha {\bf1}_q $;
\item Model 2 (heavy-tailed): $ \bX_q \sim t_{q}( \bf0, \boldsymbol{\Sigma}, \nu ) $, where the degrees of freedom are $ \nu $;
\item Model 3 (skewed and heavy-tailed): $ \bX_q \sim {\cal ST}_q( \bf0, \boldsymbol{\Sigma}, \balpha, \nu ) $, where $ \balpha = (1 / \nu) {\bf1}_q $.
\end{itemize}

It can be seen that the class of distributions in Model 1 is skewed-Gaussian distributions, while in Model~2, the class of distributions is symmetric heavy-tailed. In Model 3, the non-Gaussian distributions are both skewed and heavy-tailed.

We fix the sample size $ n = 200 $, $ p = 5 $ and $ q = 2 $. Then, we vary the values of $ \alpha $ or $ \nu $ to investigate the changes in power of the tests in the distributions. The plots of the estimated power curves of the tests are presented in \autoref{fig:powerplots1}. In the panel of Model 1, the kurtosis-based tests, namely MaxK test, the Mardia kurtosis (MK) test, the $ \text{KTO}_\text{K} $ test, the MAK test and the KK test are not included, as Model 1 is concerned with skewness only. Similarly, in the panel for Model 2, the skewness-based tests, namely the MaxS test, the Mardia skewness (MS) test, the $ \text{KTO}_\text{S} $ test, the MRS test and the MAS test are not included, as Model 2 is concerned with kurtosis only. However, in the plot for Model 3, all the tests are included.

From \autoref{fig:powerplots1}, it can be seen that in Model 1, the performances of MaxS and MaxSK tests are significantly better than all other tests. The power of the MaxSK test is slightly lower than the MaxS test. This is because MaxSK combines the MaxS and the MaxK tests using Bonferroni correction. If one of the tests does not exhibit a high power, then the power of MaxSK would be lower than the best performing test. Similar observations can be made in the panel for Model 2, where the estimated powers of the MaxK test and the MaxSK test are found to be better than all other tests, and the power curve of the MaxSK test is slightly below the power curve of the MaxK test. In Model 3, the estimated powers of the MaxSK test and the MaxK test are higher than other testing procedures, while the power exhibited by the MaxS test is lower in this model (this is because the skewness in Model 3 is rather weak).
In all the cases, it can be clearly seen that our tests significantly outperform the other tests most of the time. The test for Gaussianity, i.e., the MaxSK test, always performs quite well and exhibits better performance than all the other tests.

Among the other testing procedures, the MAS test and the MAK test involve finding the projections which maximize the univariate skewness and kurtosis, respectively. Skewness-based projection pursuit has also been studied in \cite{loperfido2018skewness}. The projection pursuit methods yield directions to project the data which would maximize the skewness or kurtosis. However, if the skewness or kurtosis is supported on only a sub-dimension of the data, the obtained direction vector may not directly help in identifying the particular sub-dimension. It is also notable that the two tests MAS and MAK are found to be less powerful compared to our proposed tests in the above simulation study of estimated powers. The $ \text{KTO}_\text{K} $ test proposed in \cite{kankainen2007tests} and implemented using the \textsf{R} package \texttt{ICS} \cite{ICS} involves the ratio of the regular
covariance matrix and the matrix of fourth moments, which is closely related to independent component analysis (see, e.g., section 2.4 in \cite{oja2006scatter}). In independent component analysis, one tries to find a linear transform of the original multivariate data vectors so that the components of the transformed vectors are independent. However, the methodology of independent component analysis would not help in detecting the sub-dimension supporting the skewness or kurtosis present in the data due to the effect of taking linear transforms of the original data. The $ \text{KTO}_\text{K} $ test is also found to be less powerful compared to our proposed tests in the above simulation study.

In \cite{franceschini2019maxskew}, the authors conjectured that the normality test based on Mardia skewness is less powerful when skewness is present in a lower-dimensional space than the space of the underlying distribution, where the lower-dimensional space may be a few components of the random vector following the underlying distribution, or a projection to a lower-dimensional space. In the plots of the estimated powers, this phenomenon is clearly observed with the Mardia tests exhibiting significantly less power compared to our proposed tests and other tests in models where the non-Gaussian features are supported on only a small sub-dimension of the overall distribution. In \cite{franceschini2019maxskew}, graphical methods to investigate skewness in such scenarios are described using skewness-based projection pursuit.

In section 2 of the supplementary material, estimated powers of the tests are presented for the smaller sample size $ n = 50 $. For this smaller sample size, it is found that in Model 1, which is the skewed model, the powers of all the tests decrease considerably, and a few tests exhibit higher powers compared to the MaxSK test. However, the power of the MaxS test is found to be higher than all these tests, except the R test. The R test sometimes exhibits higher estimated power compared to the MaxS test. In Model 2, which is the heavy-tailed model, again it is found that the estimated power of the R test is slightly higher than the MaxK and the MaxSK tests, although they are very close to each other. The other tests exhibit lower powers, but some of them are very close to the powers of the MaxK and MaxSK tests. In Model 3, which is the skewed and heavy-tailed model, it is found that the powers of the MaxS and MaxSK tests are considerably higher than all other tests. The power of the MaxK test is very close to that of the R test, which is significantly lower than the powers of the MaxS test and the MaxSK test.

In section 3 of the supplementary material, the estimated powers of the tests are also presented in an equicorrelation model with high pairwise correlation between the components of the vector $ \mathbf{X}_i $, for sample size $ n = 200 $ and $ 50 $. This experiment is carried out to evaluate the effect of the high pairwise correlation on the powers of the tests, and how this effect varies with sample size. It is found that for $ n = 200 $, the estimated powers of the MaxS, MaxSK and R tests increase sharply in Model 1, which is the skewed model, under high pairwise correlation. There is almost no difference among the powers of the MaxS, MaxSK and R tests in Model 1. However, in Model 2 and Model 3, the estimated powers of the R test is significantly lower than several other tests, and the estimated powers of the MaxK test and the MaxSK test are significantly higher than all other tests. On the other hand, when $ n = 50 $, it is found that in Model 1 under high pairwise correlation, the power of the R test is significantly higher than all of tests, and the powers of the MaxS and MaxSK tests are lower than that of the R test but significantly higher than the powers of all of the other tests. This strong dominance of the power of the R test is not maintained in case of Model 2 and Model 3 under high pairwise correlation, where several other tests exhibit higher powers compared to the R test. There, the MaxK and MaxSK tests exhibit the highest powers, but the powers of several other tests are close to them.

\subsection{Detection of Sub-Dimensions Supporting Skewness and Excess Kurtosis}\label{subsec:simulation-3}

The testing procedures described earlier can be used to detect the sub-dimensions supporting non-Gaussian features in the data. Suppose the data are skewed, but skewness is supported only on a small sub-dimension of the data. Then, to detect the sub-dimension supporting skewness, we can first conduct the MaxS test. If the p-value of the MaxS test is small (say, lower than 5\%), then there is statistical evidence of presence of skewness in the sample. Next, we find the sub-dimension corresponding to the maximum $ \tilde{b}_{1,q,i} $, which is the detected sub-dimension supporting skewness in the data. Similarly, if a heavy-tailed component is present in a small sub-dimension of the data, and we wish to detect that sub-dimension, we can conduct the MaxK test and if it rejects Gaussianity, we find the sub-dimension which corresponds to the maximum $| \tilde{b}_{2,q,i}| $. If we want to detect the sub-dimension supporting a non-Gaussian distribution, we can use the MaxSK test in the following way. We first conduct the MaxSK test. If it rejects Gaussianity, then we find which p-value, whether for MaxS or MaxK, caused the rejection. If it is only one of the tests, say MaxS, then we detect the sub-dimension corresponding to the maximum $ \tilde{b}_{1,q,i} $. Otherwise, if the p-values of both the MaxS test and the MaxK test are below 2.5\% (due to Bonferroni correction on the 5\% nominal level), then we find the sub-dimensions corresponding to the maximum $ \tilde{b}_{1,q,i} $ and the maximum $ | \tilde{b}_{2,q,i}| $. The union of these two sub-dimensions is the detected sub-dimension supporting the non-Gaussian distribution.

\begin{table}[b!]
\begin{center}
\caption{Indices of all sub-dimensions for $ p = 5 $.}
\label{tab:indices}
\begin{tabular} {r c | r c | r c}  
 \hline 
Sub-dimension   &Index  &Sub-dimension  &Index  &Sub-dimension  &Index  \\\hline
(1)             &1      &(2)            &2      &(3)            &3      \\
(4)             &4      &(5)            &5      &(1, 2)         &6      \\
(1, 3)          &7      &(1, 4)         &8      &(1, 5)         &9      \\
(2, 3)          &10     &(2, 4)         &11     &(2, 5)         &12     \\
(3, 4)          &13     &(3, 5)         &14     &(4, 5)         &15     \\
(1, 2, 3)       &16     &(1, 2, 4)      &17     &(1, 2, 5)      &18     \\
(1, 3, 4)       &19     &(1, 3, 5)      &20     &(1, 4, 5)      &21     \\
(2, 3, 4)       &22     &(2, 3, 5)      &23     &(2, 4, 5)      &24     \\
(3, 4, 5)       &25     &(1, 2, 3, 4)   &26     &(1, 2, 3, 5)   &27     \\
(1, 2, 4, 5)    &28     &(1, 3, 4, 5)   &29     &(2, 3, 4, 5)   &30     \\
(1, 2, 3, 4, 5) &31     &               &       &               &       \\\hline
\end{tabular}
\end{center}
\end{table}

\begin{figure}[b!]
\centering
\includegraphics[width=0.67\linewidth]{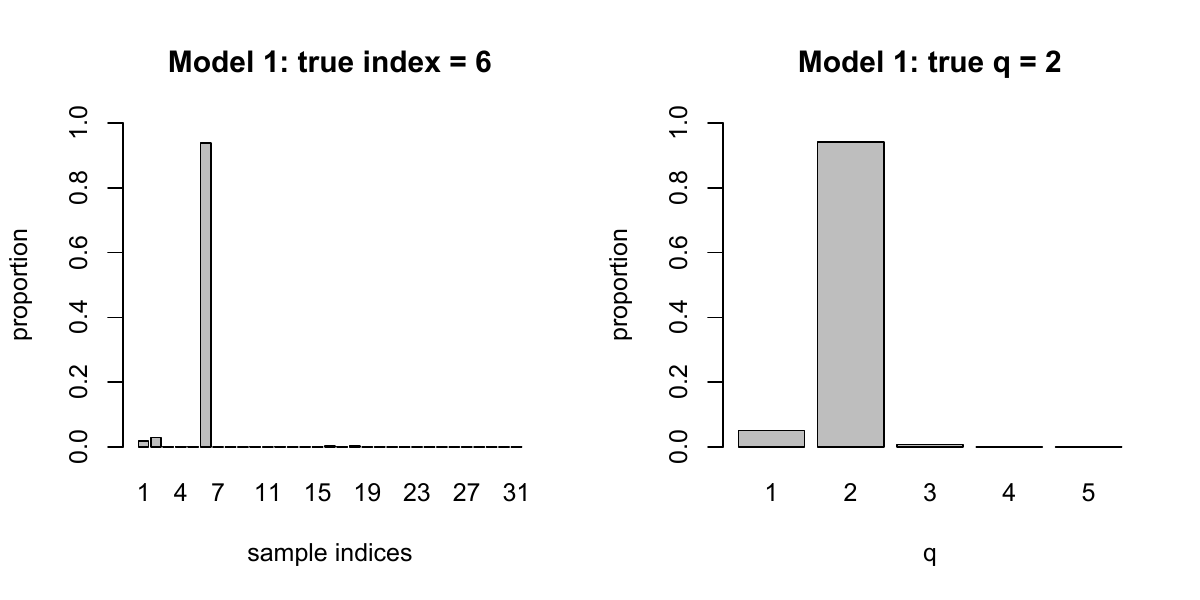}
\includegraphics[width=0.67\linewidth]{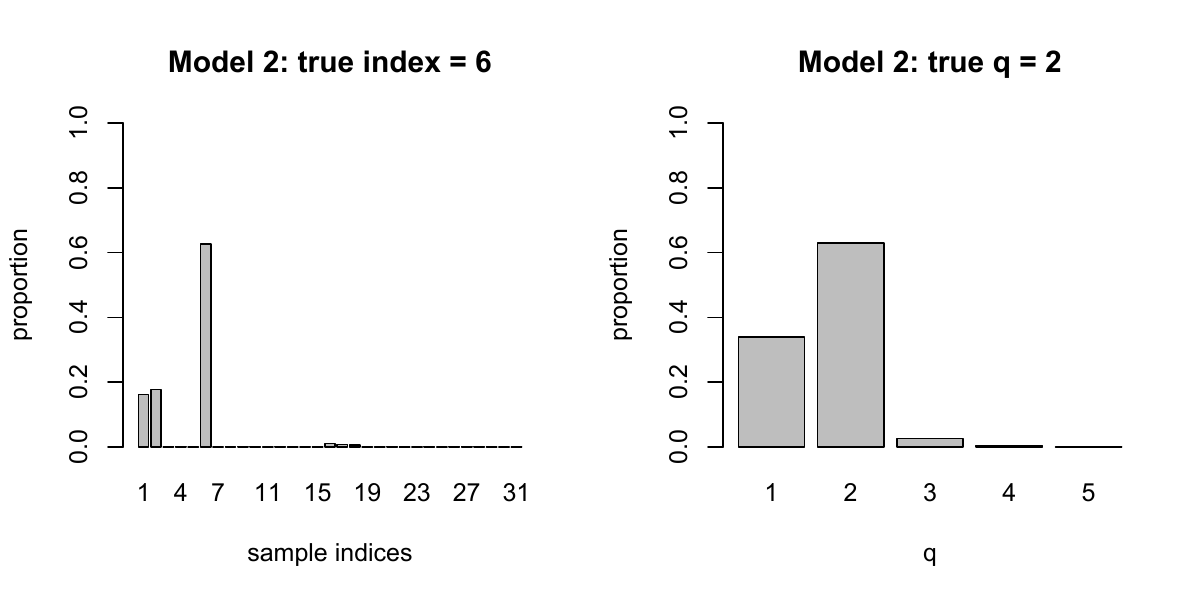}
\includegraphics[width=0.67\linewidth]{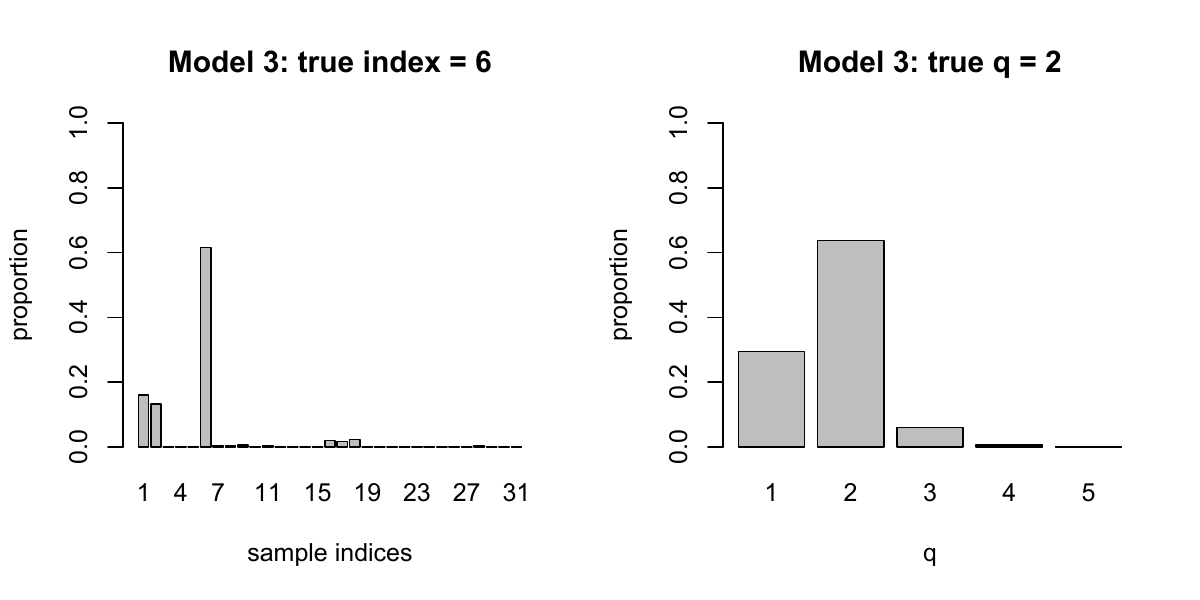}
\caption{Histograms of the detection rate of sub-dimensions in Model 1 using the centered and scaled skewness measure (first row), in Model 2 using the centered and scaled kurtosis measure (second row), in Model 3 using the procedure for detecting non-Gaussianity (third row). In all the cases, $ n = 200 $, $p=5$ and the true sub-dimension supporting skewness or excess kurtosis is $q=2$. Results based on 1000 replicates.}
\label{fig:skewkurtdimdetect}
\end{figure}
To investigate the performance of the detection procedure described above, we consider the three models described in \autoref{subsec:simulation-2}. In Model 1, we fix $ \alpha = 5 $ and conduct the detection procedure 1000 times on independent replicates. The proportion of times each of the sub-dimensions is detected as the one supporting the skewed distribution is computed. Also, the size of the sub-dimensions (denoted as $ q $) thus detected to support the skewed distribution is also recorded, and the proportions for the $ q $ values are computed. All the possible sub-dimensions from $ p = 5 $ variables are assigned indices, which are presented in \autoref{tab:indices}, and the proportions thus computed are plotted against the indices in the first row of \autoref{fig:skewkurtdimdetect}. The estimated power of the MaxS test there is 0.922. It can be clearly seen that the highest proportions in the respective histograms are attained for the true sub-dimension and the true $ q $.

Similarly, the procedure to detect the sub-dimension supporting the heavy-tailed distribution using the MaxK test is also carried out based on 1000 replicates in Model 2 with $ \nu = 5 $. The estimated power of the MaxK test there is 0.985. The histograms are presented in the second row of \autoref{fig:skewkurtdimdetect}, and we again see that the highest proportions are attained for the true sub-dimension and true $ q $.

Finally, the procedure to detect the sub-dimension supporting a non-Gaussian distribution using the MaxSK test is carried out in Model 3 taking $ \nu = 1/\alpha=5 $ and based on 1000 independent replicates. The estimated power of the MaxSK test is 0.974 there. The histograms are presented in the third row of \autoref{fig:skewkurtdimdetect}, where we again see that the highest proportions are attained for the true sub-dimension and true value of $ q $.

\section{Sub-Dimensional Data Analysis}\label{sec:realdata}
\subsection{Fisher's Iris Data}\label{subsec:iris}

We revisit Fisher's iris dataset discussed in \autoref{sec:intro}, where we considered a part of the data related to the species `iris setosa' to demonstrate that the Mardia test of skewness fails to detect skewed features in sub-dimensions. In \autoref{tab:realdatapvalues}, the p-values of all the tests are presented. It can be seen that our test detects skewness in the data, while the Mardia test fails.
Similar observations were made in \cite{MVN}.
The p-values of the $ \text{MaxS}_1 $, $ \text{MaxS}_2 $ and $ \text{MaxSK}_1 $ tests, the R test, the DH test, the VG test and the ZS test are also smaller than the 5\% nominal level.
Next, we consider the whole Fisher's iris dataset and compute the p-values of all the tests. We find that our test of kurtosis can detect the deviation of kurtosis from Gaussian kurtosis, while the Mardia test of kurtosis fails.
The p-values of the majority of the other tests are also smaller than the 5\% nominal level.

Fisher's iris dataset is generally modeled using a Gaussian distribution. However, our findings point to the non-Gaussianity of the data, and thus it may be judicious to use non-Gaussian and skewed distributions while analyzing this dataset.

\begin{table}[t!]
\begin{center}
\caption{Estimated p-values of the tests for the two data examples. Significant p-values at 5\% level are in bold.}
\label{tab:realdatapvalues}
\begin{tabular}{c|cc|c}
\hline
                    &Iris setosa    &Iris           &Wind       \\\hline
Test name           &$p=4$          &$p=4$          &$p=3$          \\\hline
MS         			&0.236 			&\textbf{0.000} &\textbf{0.004} \\
MK         			&0.448 			&0.611          &0.338          \\\hline
MaxS   			    &\textbf{0.001} &\textbf{0.000}	&\textbf{0.032} \\
MaxK  			    &0.360 			&\textbf{0.003} &\textbf{0.007} \\
MaxSK  			    &\textbf{0.002}	&\textbf{0.000} &\textbf{0.014} \\\hline
$ \text{MaxS}_1 $   &\textbf{0.002}	&0.317          &\textbf{0.023} \\
$ \text{MaxS}_2 $   &\textbf{0.047}	&\textbf{0.000} &\textbf{0.011} \\
$ \text{MaxS}_3 $   &0.132 			&\textbf{0.000} &\textbf{0.001} \\
$ \text{MaxS}_4 $   &0.235 			&\textbf{0.000} &--	            \\
$ \text{MaxK}_1 $   &0.244 			&\textbf{0.000} &\textbf{0.004}	\\
$ \text{MaxK}_2 $   &0.236 			&0.068          &0.124          \\
$ \text{MaxK}_3 $   &0.570 			&0.258          &0.367          \\
$ \text{MaxK}_4 $   &0.436 			&0.568          &--		        \\
$ \text{MaxSK}_1 $  &\textbf{0.004}	&\textbf{0.000} &\textbf{0.008} \\
$ \text{MaxSK}_2 $  &0.094 			&\textbf{0.000} &\textbf{0.022} \\
$ \text{MaxSK}_3 $  &0.264   		&\textbf{0.000} &\textbf{0.002} \\
$ \text{MaxSK}_4 $  &0.470 			&\textbf{0.000} &--             \\\hline
HZ         			&0.050      	&\textbf{0.000} &0.097          \\
R          			&\textbf{0.000}	&\textbf{0.000} &\textbf{0.011} \\
DH         			&\textbf{0.000}	&\textbf{0.000} &\textbf{0.022} \\
$ \text{KT}_S $ 	&0.221 			&\textbf{0.040} &0.710          \\
$ \text{KT}_K $ 	&0.875			&\textbf{0.046} &0.349          \\
BS         			&0.060			&\textbf{0.000} &\textbf{0.000} \\
VG         			&\textbf{0.012}	&\textbf{0.000} &0.057          \\
ZS         			&\textbf{0.020}	&0.090          &0.090          \\
MRS                 &0.212          &\textbf{0.038} &0.708          \\
MAS                 &0.342          &\textbf{0.034} &0.142          \\
MAK                 &0.258          &0.906          &0.064          \\
KK                  &0.060          &0.312          &0.080          \\\hline
\end{tabular}
\end{center}
\end{table}

\subsection{Wind Speed Data in Saudi Arabia}\label{subsec:wind}   % The wind data is (loc\_nn[,3], loc\_GG[,1], loc\_GG[,2]).

We consider a trivariate windspeed dataset produced by Yip \cite{Yip2018} with the Weather Research and Forecasting (WRF) model. The three components correspond to bi-weekly mid-day windspeed during the period 2009-2014 at three locations near Dumat Al Jandal, the first wind farm currently under construction in Saudi Arabia. It is important to study the distributional properties of this trivariate windspeed vector because they are crucial for understanding wind patterns that will influence the production of electricity by the nearby wind farm. In particular, it is of interest to assess whether a Gaussian distribution is suitable, or a non-Gaussian model needs to be developed.

The dataset consists of $n=156$ trivariate windspeed vectors. A Ljung-Box test reveals no indication of serial dependence, hence, the dataset is treated as a random sample from a three-dimensional distribution.  The p-values of the various tests are listed in \autoref{tab:realdatapvalues}. At the 5\% level, the Mardia tests support skewness but do not reject a Gaussian kurtosis. Our global tests, however, reject both symmetry and Gaussian kurtosis, suggesting that a non-Gaussian distribution would be more suitable to model these data. Looking at our sub-dimensional tests, we observe that skewness is rejected in all sub-dimensions, whereas Gaussian kurtosis is rejected only for the $q=1$ dimensional marginals. Among the other twelve tests of Gaussianity, only three reject Gaussianity whereas the remaining nine cannot. 

In summary, our new tests suggest to use a non-Gaussian distribution to model these data. They provide additional information about the non-Gaussian behavior in sub-dimensional components of the trivariate distribution.

\section{Discussion}\label{sec:conclusion}
We have developed some new tests of skewness and kurtosis which take into account the skewness and excess kurtosis present in the sub-dimensions of the data. It was demonstrated through analyses of simulated and real data that our tests outperform the classical Mardia tests of skewness and kurtosis when the skewness and the excess kurtosis are present in a small sub-dimension of the variables under consideration. Moreover, our tests can also be used as tests of Gaussianity, and it was observed that as such, they outperform several popular tests of Gaussianity. We have further developed a methodology to detect the true sub-dimension when the skewness and the excess kurtosis are supported on a small sub-dimension of the data.

One limitation of our methodology is that it considers all the possible sub-dimensions, which is $2^p-1$, to detect skewness or excess kurtosis. The number $2^p-1$ becomes large for even moderate values of $ p $. So, the methodology is computationally  intensive. Future research needs to develop suitable computational methods when the dimension $p$ of the multivariate data is high to reduce the computation burden. In particular, methodology is required to be developed when the data are high-dimensional in nature, i.e., $p > n$. We discuss some possible ways.

In the high-dimensional setup, in case it is known that the skewness or the non-Gaussian kurtosis can only possibly be supported on a sub-dimension $q_0$ (with $q_0$ relatively small compared to the sample size $n$), then the testing procedures $ \text{MaxS}_{q_0} $, $ \text{MaxK}_{q_0} $ and $ \text{MaxSK}_{q_0} $ can be applied, which are described in \autoref{subsec:skewnesstest}, \autoref{subsec:kurtosistest} and \autoref{subsec:gaussianitytest}, respectively. However, in a setup where such information is not available, new procedures need to be developed. One way is to randomly select a fixed number of sub-dimensions from the collection of $2^p - 1$ possible sub-dimensions, carry out the usual Mardia tests of skewness and kurtosis on those sub-dimensions, and then combine the results of those Mardia tests based on a multiple testing procedure to get the result of the overall test.
A second possible way is to consider a fixed but suitably large number of random projections from the original dimension $p$ to a smaller dimension $p'$, where our procedures can be applied. Then, our testing procedures $ \text{MaxS} $, $ \text{MaxK} $ and $ \text{MaxSK} $ can be applied on the $p'$-dimensional projected data for all the random projections, and the results for all the random projections can be combined using a multiple testing method.
A third possible way can be to perform principal component analysis to reduce the dimension of the data from $p$ to $p'$, and then apply our procedures on the dimension-reduced data.

The procedures proposed herein are developed based on the Mardia measures and tests, and these tests cannot detect non-Gaussianity if none of the sub-dimensional Mardia measures are able to detect it. One example of such a non-Gaussian distribution is given in \cite{dutta2014non}. There, all the sub-dimensional Mardia measures including the global Mardia measures coincide with the Gaussian distribution. In such a case, our methodology will not work, and different methods are required.

\section*{Supplementary material} The supplementary material contains some additional simulations: estimated sizes for $ p = 3 \text{ and } 4 $, estimated powers for $ n = 50 $, and investigation of the performance of the tests under high pairwise correlation.

\section*{Acknowledgements} This research was supported by the King Abdullah University of Science and Technology (KAUST). We thank the associate editor and the anonymous reviewers for their helpful comments and suggestions.

\bibliographystyle{myjmva}
\bibliography{bibliography}

\end{document}

\begin{frontmatter}

\title{Supplementary material:\\ Sub-Dimensional Mardia Measures of Multivariate Skewness and Kurtosis}

\author[mymainaddress1]{Joydeep Chowdhury\corref{mycorrespondingauthor}}
\author[mymainaddress2]{Subhajit Dutta}
\author[mymainaddress3]{Reinaldo B. Arellano-Valle}
\author[mymainaddress1]{Marc~G.~Genton}
\address[mymainaddress1]{Statistics Program, King Abdullah University of Science and Technology, Thuwal, Saudi Arabia}
\address[mymainaddress2]{Department of Mathematics and Statistics, Indian Institute of Technology, Kanpur 208016, India}
\address[mymainaddress3]{Department of Statistics, Pontificia Universida Cat\'olica de Chile, Santiago 22, Chile}

\cortext[mycorrespondingauthor]{Corresponding author}

\end{frontmatter}

\section{Estimated sizes for different $ n $ and $ p $}\label{sec:sizes}
In this section, we present the estimated sizes of our tests for different sample size $ n $ and dimension $ p $ in some simulated models. As in the main paper, for the estimation of the sizes, we consider $ \mathbf{X}_1, \ldots, \mathbf{X}_n $ being a random sample from the $ p $-variate Gaussian distribution $ {\cal N}_p( \mathbf{0}, \boldsymbol{\Sigma} ) $, where $ \boldsymbol{\Sigma} = ( \sigma_{ij} ) $ with $ \sigma_{ij} = 0.5 + 0.5 \mathbb{I}( i = j ) $. We consider the cases $ p = 3 \text{ and } 4 $. The number of replicates to estimate the sizes of our tests is taken as $ 1000 $. The sample size $ n $ is varied. The estimated sizes for $ p = 3 $ are presented in \autoref{tab:size_p3}, and those for $ p = 4 $ are presented in \autoref{tab:size_p4}.

\begin{table}[h!]
\begin{center}
\caption{Estimated sizes of the tests for 5\% nominal level in $ {\cal N}_3( \mathbf{0}, \boldsymbol{\Sigma} ) $ based on 1000 replicates, where $ \boldsymbol{\Sigma} = ( \sigma_{ij} ) $ with $ \sigma_{ij} = 0.5 + 0.5 \mathbb{I}( i = j ) $.}
\label{tab:size_p3}
\begin{tabular} {l c c c c c}  
\hline 
Test & $ n = 50 $ & $ n = 100 $ & $ n = 200 $ & $ n = 500 $ & $ n = 1000 $ \\ \hline 
MaxS & 0.046 & 0.053 & 0.050 & 0.047 & 0.061 \\ 
MaxK & 0.038 & 0.044 & 0.040 & 0.046 & 0.058 \\ 
MaxSK & 0.043 & 0.064 & 0.047 & 0.050 & 0.064 \\ 
$ \text{MaxS}_1 $ & 0.047 & 0.048 & 0.048 & 0.049 & 0.054 \\ 
$ \text{MaxS}_2 $ & 0.047 & 0.051 & 0.058 & 0.051 & 0.058 \\ 
$ \text{MaxS}_3 $ & 0.034 & 0.037 & 0.052 & 0.045 & 0.058 \\ 
$ \text{MaxK}_1 $ & 0.042 & 0.045 & 0.042 & 0.043 & 0.050 \\ 
$ \text{MaxK}_2 $ & 0.022 & 0.021 & 0.028 & 0.037 & 0.052 \\ 
$ \text{MaxK}_3 $ & 0.040 & 0.042 & 0.038 & 0.049 & 0.070 \\ 
$ \text{MaxSK}_1 $ & 0.044 & 0.065 & 0.054 & 0.055 & 0.051 \\ 
$ \text{MaxSK}_2 $ & 0.026 & 0.033 & 0.038 & 0.039 & 0.061 \\ 
$ \text{MaxSK}_3 $ & 0.026 & 0.034 & 0.043 & 0.042 & 0.071 \\ \hline 
\end{tabular}
\end{center}
\end{table}

\begin{table}[h!]
\begin{center}
\caption{Estimated sizes of the tests for 5\% nominal level in $ {\cal N}_4( \mathbf{0}, \boldsymbol{\Sigma} ) $ based on 1000 replicates, where $ \boldsymbol{\Sigma} = ( \sigma_{ij} ) $ with $ \sigma_{ij} = 0.5 + 0.5 \mathbb{I}( i = j ) $.}
\label{tab:size_p4}
\begin{tabular} {l c c c c c}  
\hline 
Test & $ n = 50 $ & $ n = 100 $ & $ n = 200 $ & $ n = 500 $ & $ n = 1000 $ \\ \hline 
MaxS & 0.045 & 0.057 & 0.054 & 0.045 & 0.050 \\ 
MaxK & 0.043 & 0.053 & 0.059 & 0.051 & 0.051 \\ 
MaxSK & 0.046 & 0.058 & 0.063 & 0.054 & 0.059 \\ 
$ \text{MaxS}_1 $ & 0.039 & 0.057 & 0.065 & 0.049 & 0.042 \\ 
$ \text{MaxS}_2 $ & 0.047 & 0.055 & 0.053 & 0.035 & 0.060 \\ 
$ \text{MaxS}_3 $ & 0.038 & 0.047 & 0.050 & 0.034 & 0.062 \\ 
$ \text{MaxS}_4 $ & 0.029 & 0.043 & 0.045 & 0.032 & 0.051 \\ 
$ \text{MaxK}_1 $ & 0.040 & 0.053 & 0.063 & 0.044 & 0.048 \\ 
$ \text{MaxK}_2 $ & 0.022 & 0.038 & 0.042 & 0.038 & 0.045 \\ 
$ \text{MaxK}_3 $ & 0.019 & 0.023 & 0.038 & 0.041 & 0.036 \\ 
$ \text{MaxK}_4 $ & 0.088 & 0.068 & 0.058 & 0.042 & 0.040 \\ 
$ \text{MaxSK}_1 $ & 0.044 & 0.059 & 0.068 & 0.054 & 0.057 \\ 
$ \text{MaxSK}_2 $ & 0.030 & 0.051 & 0.051 & 0.039 & 0.052 \\ 
$ \text{MaxSK}_3 $ & 0.023 & 0.038 & 0.039 & 0.040 & 0.047 \\ 
$ \text{MaxSK}_4 $ & 0.044 & 0.044 & 0.047 & 0.035 & 0.050 \\ \hline 
\end{tabular}
\end{center}
\end{table}

In \autoref{tab:size_p3}, which corresponds to $ p = 3 $, it can be seen that most of the tests have estimated sizes close to the nominal level for $ n = 50 $ and up, while the tests $ \text{MaxK}_2 $ and $ \text{MaxSK}_2 $ require higher sample sizes for their estimated sizes to reach close to the nominal level.

In \autoref{tab:size_p4}, which corresponds to $ p = 4 $, it can be seen that again the majority of the tests have estimated sizes close to the nominal level for $ n = 50 $ and up. However, the tests $ \text{MaxS}_4 $, $ \text{MaxK}_2 $, $ \text{MaxK}_3 $, $ \text{MaxK}_4 $ and $ \text{MaxSK}_3 $ require higher sample sizes for their estimated sizes to be close to the nominal level of 5\%.

\section{Estimated powers for smaller $ n $}\label{sec:powers}
While our testing procedures are asymptotic in nature, it is of interest to see how they perform when the sample size is not large. For this, the simulation experiment of estimated powers presented in section 7.2 of the main paper is carried out here keeping everything else same, but setting the sample size $ n = 50 $. The curves of estimated powers are presented in \autoref{fig:powerplotsn50}.

\begin{figure}[h!]
\centering
\includegraphics[height=5.9cm]{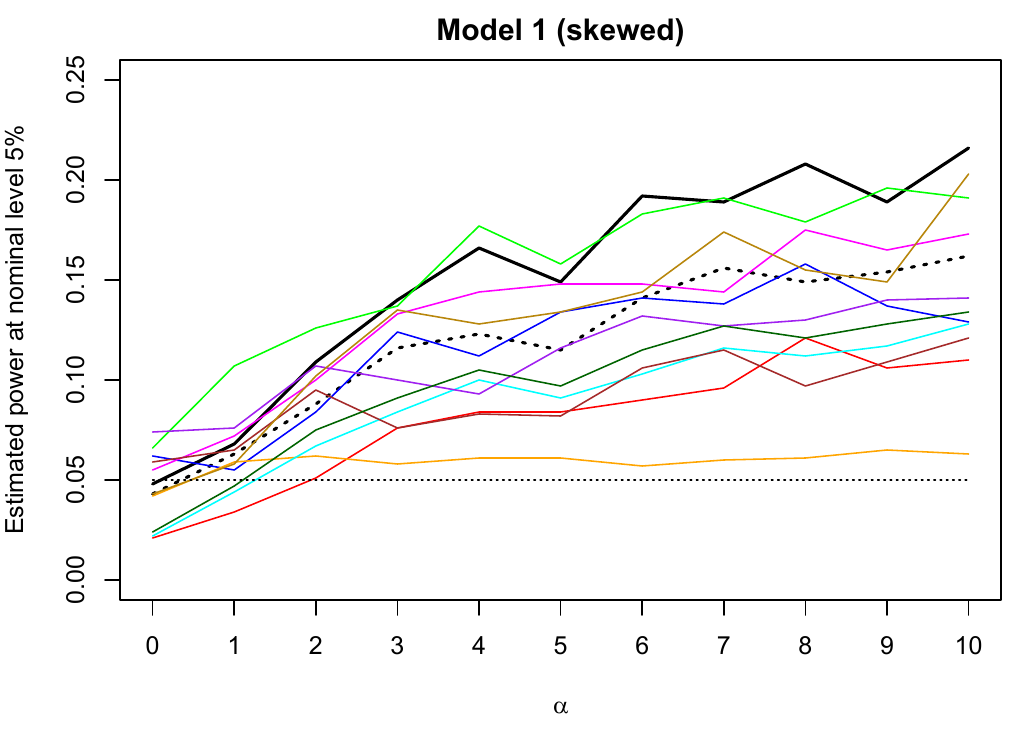}
\includegraphics[height=5.9cm]{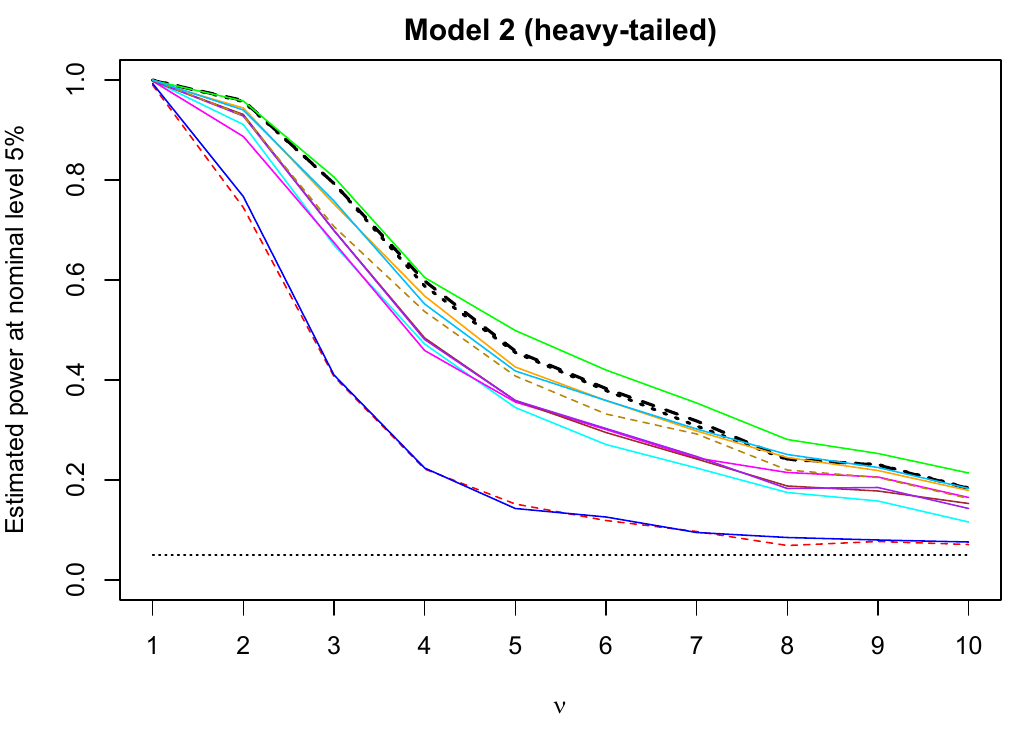}
\includegraphics[height=5.9cm]{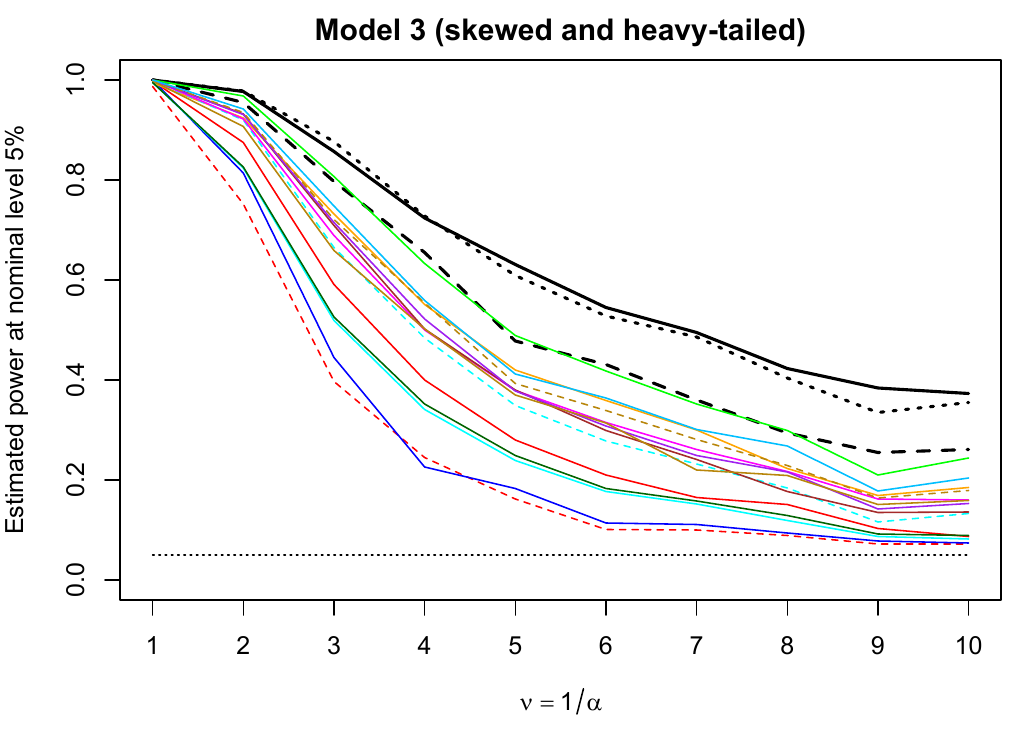}
\includegraphics[width=1\linewidth,trim={0 70 0 230},clip]{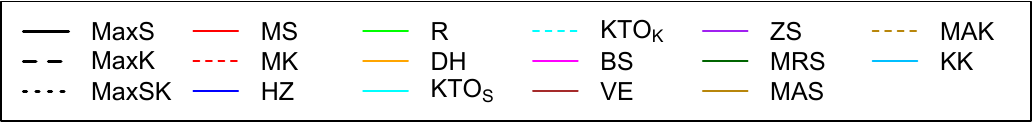}
\caption{Estimated powers of MaxS, MaxK, MaxSK, MS, MK, HZ, R, DH, $ \text{KTO}_\text{S} $, $ \text{KTO}_\text{K} $, BS, VE, ZS, MRS, MAS, MAK and KK tests for 5\% nominal level (horizontal dashed line) in Model 1 (top left), Model 2 (top right) and Model 3 (bottom) for $ n = 50 $, $p=5$, $q=2$ based on 1000 replicates. The horizontal dotted lines near the bottom of the plots correspond to the nominal level of 5\%.}
\label{fig:powerplotsn50}
\end{figure}

It can be seen that the powers of the tests have substantially decreased in \autoref{fig:powerplotsn50} compared to those in Fig. 2 of the main paper, which is expected due to the smaller sample size. The decrease in power is highest in Model 1, where the highest estimated power is not even 0.25. The MaxS test is found to be either the best performing test or close to the best performing test in Model 1, although several tests exhibit higher powers compared to the MaxSK test there. The R test is either the best performing test or has estimated power very close to that of the MaxS test depending on the value of the parameter $ \alpha $ in Model 1. In Model 2, the estimated power of the R test is almost uniformly better than all other tests, although the powers of the MaxK test and the MaxSK test are always very close to the best. In Model 3, the estimated powers of the MaxS test and the MaxSK test are substantially higher than the other tests, including the MaxK test. It is notable that in Fig. 2 of the main paper, the estimated power of the MaxK test was higher than that of the MaxS test in Model 3, which is the opposite of what is observed in \autoref{fig:powerplotsn50}. All these observations indicate that the performance of the tests change with increasing sample sizes, and that the gaps between the performances of the MaxS, MaxK and MaxSK tests and those of the other tests increase as the sample size increases.

\section{Performance of the tests under high correlation}\label{sec:sizes_highcor}
It is of interest to investigate how the performances of the tests are affected when the underlying distribution is such that the pairwise correlation between the components is high. To investigate the effect of this high pairwise correlation on the estimated sizes of the tests, the same setup as in \autoref{sec:sizes} is considered with the exception that here $ \sigma_{ij} = 0.99 + 0.01 \mathbb{I}( i = j ) $. The sizes of the tests are estimated for different values of $ n $ and for the cases $ p = 3, 4 \text{ and } 5 $.

\begin{table}[!h]
\begin{center}
\caption{Estimated sizes of the tests for 5\% nominal level in $ {\cal N}_3( \mathbf{0}, \boldsymbol{\Sigma} ) $ based on 1000 replicates, where $ \boldsymbol{\Sigma} = ( \sigma_{ij} ) $ with $ \sigma_{ij} = 0.99 + 0.01 \mathbb{I}( i = j ) $.}
\label{tab:size_p3c99}
\begin{tabular} {l c c c c c}  
\hline 
Test & $ n = 50 $ & $ n = 100 $ & $ n = 200 $ & $ n = 500 $ & $ n = 1000 $ \\ \hline 
MaxS & 0.044 & 0.044 & 0.037 & 0.051 & 0.050 \\ 
MaxK & 0.029 & 0.041 & 0.046 & 0.051 & 0.039 \\ 
MaxSK & 0.036 & 0.036 & 0.048 & 0.053 & 0.049 \\ 
$ \text{MaxS}_1 $ & 0.040 & 0.047 & 0.043 & 0.052 & 0.048 \\ 
$ \text{MaxS}_2 $ & 0.049 & 0.043 & 0.037 & 0.055 & 0.043 \\ 
$ \text{MaxS}_3 $ & 0.039 & 0.030 & 0.036 & 0.043 & 0.044 \\ 
$ \text{MaxK}_1 $ & 0.023 & 0.029 & 0.041 & 0.039 & 0.041 \\ 
$ \text{MaxK}_2 $ & 0.017 & 0.024 & 0.036 & 0.052 & 0.044 \\ 
$ \text{MaxK}_3 $ & 0.027 & 0.040 & 0.047 & 0.050 & 0.045 \\ 
$ \text{MaxSK}_1 $ & 0.028 & 0.036 & 0.047 & 0.043 & 0.051 \\ 
$ \text{MaxSK}_2 $ & 0.030 & 0.034 & 0.039 & 0.053 & 0.043 \\ 
$ \text{MaxSK}_3 $ & 0.025 & 0.032 & 0.040 & 0.054 & 0.045 \\ \hline 
\end{tabular}
\end{center}
\end{table}

\begin{table}[!h]
\begin{center}
\caption{Estimated sizes of the tests for 5\% nominal level in $ {\cal N}_4( \mathbf{0}, \boldsymbol{\Sigma} ) $ based on 1000 replicates, where $ \boldsymbol{\Sigma} = ( \sigma_{ij} ) $ with $ \sigma_{ij} = 0.99 + 0.01 \mathbb{I}( i = j ) $.}
\label{tab:size_p4c99}
\begin{tabular} {l c c c c c}  
\hline 
Test & $ n = 50 $ & $ n = 100 $ & $ n = 200 $ & $ n = 500 $ & $ n = 1000 $ \\ \hline 
MaxS & 0.047 & 0.043 & 0.044 & 0.054 & 0.051 \\ 
MaxK & 0.025 & 0.043 & 0.045 & 0.041 & 0.053 \\ 
MaxSK & 0.036 & 0.047 & 0.048 & 0.042 & 0.044 \\ 
$ \text{MaxS}_1 $ & 0.049 & 0.042 & 0.049 & 0.062 & 0.047 \\ 
$ \text{MaxS}_2 $ & 0.048 & 0.040 & 0.048 & 0.056 & 0.046 \\ 
$ \text{MaxS}_3 $ & 0.036 & 0.043 & 0.044 & 0.054 & 0.055 \\ 
$ \text{MaxS}_4 $ & 0.028 & 0.033 & 0.047 & 0.052 & 0.048 \\ 
$ \text{MaxK}_1 $ & 0.021 & 0.035 & 0.036 & 0.034 & 0.046 \\ 
$ \text{MaxK}_2 $ & 0.017 & 0.031 & 0.041 & 0.046 & 0.044 \\ 
$ \text{MaxK}_3 $ & 0.007 & 0.036 & 0.032 & 0.059 & 0.048 \\ 
$ \text{MaxK}_4 $ & 0.070 & 0.074 & 0.041 & 0.052 & 0.045 \\ 
$ \text{MaxSK}_1 $ & 0.036 & 0.042 & 0.040 & 0.044 & 0.048 \\ 
$ \text{MaxSK}_2 $ & 0.037 & 0.044 & 0.043 & 0.051 & 0.044 \\ 
$ \text{MaxSK}_3 $ & 0.020 & 0.031 & 0.037 & 0.040 & 0.042 \\ 
$ \text{MaxSK}_4 $ & 0.026 & 0.046 & 0.050 & 0.044 & 0.047 \\ \hline 
\end{tabular}
\end{center}
\end{table}

\begin{table}[!h]
\begin{center}
\caption{Estimated sizes of the tests for 5\% nominal level in $ {\cal N}_5( \mathbf{0}, \boldsymbol{\Sigma} ) $ based on 1000 replicates, where $ \boldsymbol{\Sigma} = ( \sigma_{ij} ) $ with $ \sigma_{ij} = 0.99 + 0.01 \mathbb{I}( i = j ) $.}
\label{tab:size_p5c99}
\begin{tabular} {l c c c c c}  
\hline 
Test & $ n = 50 $ & $ n = 100 $ & $ n = 200 $ & $ n = 500 $ & $ n = 1000 $ \\ \hline 
MaxS & 0.048 & 0.059 & 0.050 & 0.057 & 0.061 \\ 
MaxK & 0.041 & 0.050 & 0.059 & 0.038 & 0.060 \\ 
MaxSK & 0.044 & 0.058 & 0.058 & 0.051 & 0.061 \\ 
$ \text{MaxS}_1 $ & 0.050 & 0.055 & 0.056 & 0.051 & 0.057 \\ 
$ \text{MaxS}_2 $ & 0.056 & 0.056 & 0.054 & 0.052 & 0.056 \\ 
$ \text{MaxS}_3 $ & 0.051 & 0.057 & 0.043 & 0.057 & 0.053 \\ 
$ \text{MaxS}_4 $ & 0.040 & 0.050 & 0.041 & 0.050 & 0.055 \\ 
$ \text{MaxS}_5 $ & 0.025 & 0.036 & 0.040 & 0.042 & 0.058 \\ 
$ \text{MaxK}_1 $ & 0.028 & 0.033 & 0.041 & 0.036 & 0.066 \\ 
$ \text{MaxK}_2 $ & 0.033 & 0.046 & 0.053 & 0.033 & 0.049 \\ 
$ \text{MaxK}_3 $ & 0.017 & 0.033 & 0.035 & 0.038 & 0.049 \\ 
$ \text{MaxK}_4 $ & 0.037 & 0.036 & 0.057 & 0.042 & 0.054 \\ 
$ \text{MaxK}_5 $ & 0.131 & 0.079 & 0.065 & 0.050 & 0.062 \\ 
$ \text{MaxSK}_1 $ & 0.040 & 0.048 & 0.054 & 0.039 & 0.066 \\ 
$ \text{MaxSK}_2 $ & 0.042 & 0.057 & 0.061 & 0.047 & 0.056 \\ 
$ \text{MaxSK}_3 $ & 0.034 & 0.047 & 0.041 & 0.050 & 0.051 \\ 
$ \text{MaxSK}_4 $ & 0.032 & 0.037 & 0.044 & 0.045 & 0.064 \\ 
$ \text{MaxSK}_5 $ & 0.069 & 0.050 & 0.054 & 0.047 & 0.057 \\ \hline 
\end{tabular}
\end{center}
\end{table}

The estimated sizes of the tests for $ p = 3, 4 \text{ and } 5 $ are presented in \autoref{tab:size_p3c99}, \autoref{tab:size_p4c99} and \autoref{tab:size_p5c99}, respectively.
It can be observed that nearly all of the tests, including the three tests MaxS, MaxK and MaxSK, have estimated sizes close to the nominal level for sample size $ n = 100 $ and higher. However, when the sample size is lower, i.e., $ n = 50 $, some deviations of the estimated sizes from the nominal level are observed. These include the estimated sizes of the tests MaxK, $ \text{MaxK}_1 $, $ \text{MaxK}_2 $, $ \text{MaxK}_3 $, $ \text{MaxSK}_1 $ and $ \text{MaxSK}_3 $ in \autoref{tab:size_p3c99}, the tests MaxK, $ \text{MaxS}_4 $, $ \text{MaxK}_1 $, $ \text{MaxK}_2 $, $ \text{MaxK}_3 $, $ \text{MaxK}_4 $, $ \text{MaxSK}_3 $ and $ \text{MaxSK}_4 $ in \autoref{tab:size_p4c99}, and the tests $ \text{MaxS}_5 $, $ \text{MaxK}_1 $, $ \text{MaxK}_3 $ and $ \text{MaxK}_5 $ in \autoref{tab:size_p5c99}. From these observations, it seems that the estimated sizes of the tests of kurtosis are somewhat affected for small sizes when the pairwise correlation in the data is high, but the tests of skewness are not affected noticeably.

Next, the effect of high pairwise correlation on the powers of the tests is investigated by considering the setup described in section 7.2 of the main paper keeping everything else same except setting $ \sigma_{ij} = 0.99 + 0.01 \mathbb{I}( i = j ) $. Two cases of sample size are considered: the first case is $ n = 200 $ as in section 7.2 of the main paper, and the second case is $ n = 50 $. The estimated powers for the first case, where $ n = 200 $, are presented in \autoref{fig:powerplotsn200eqcor99}, and those for the second case, where $ n = 50 $, are presented in \autoref{fig:powerplotsn50eqcor99}.

\begin{figure}[h!]
\centering
\includegraphics[height=5.9cm]{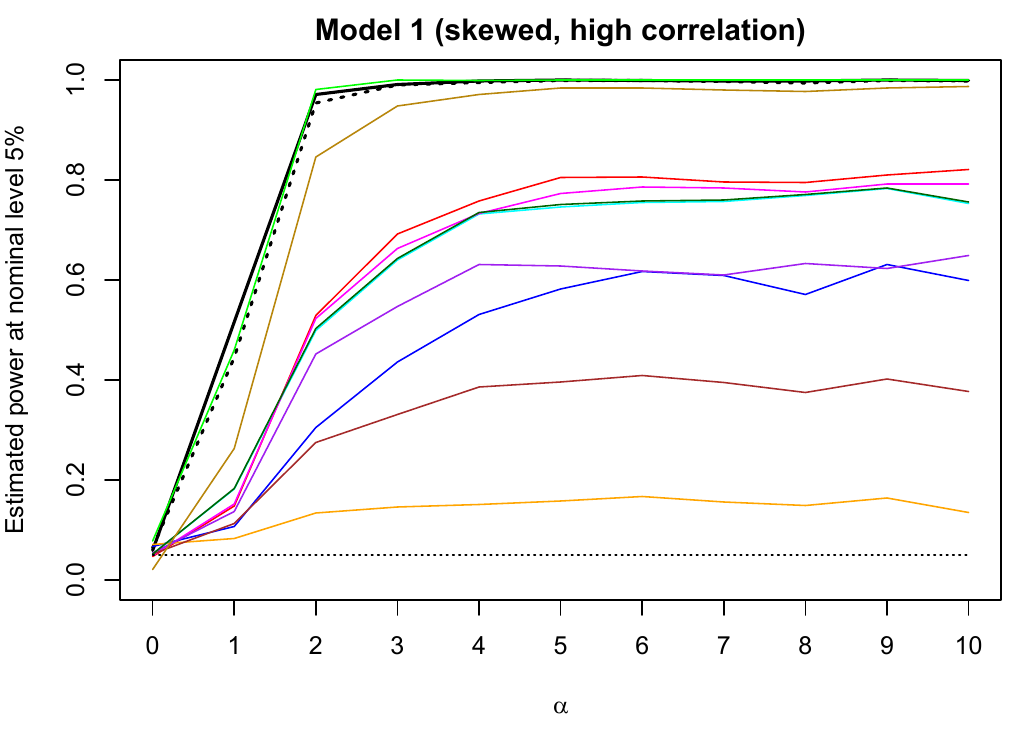}
\includegraphics[height=5.9cm]{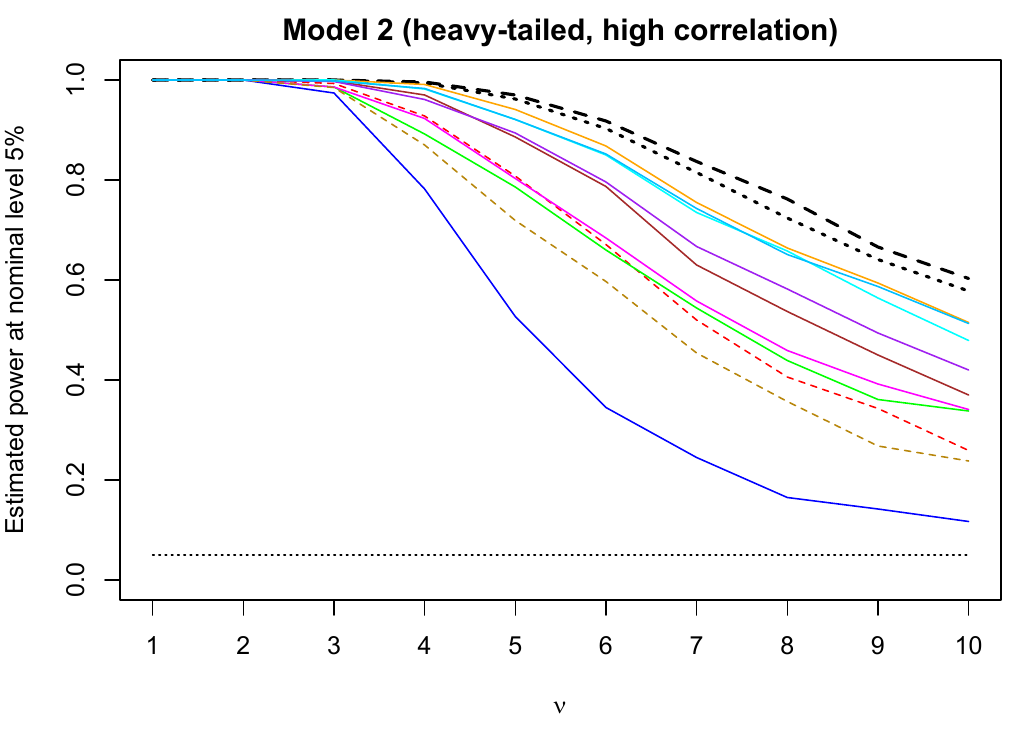}
\includegraphics[height=5.9cm]{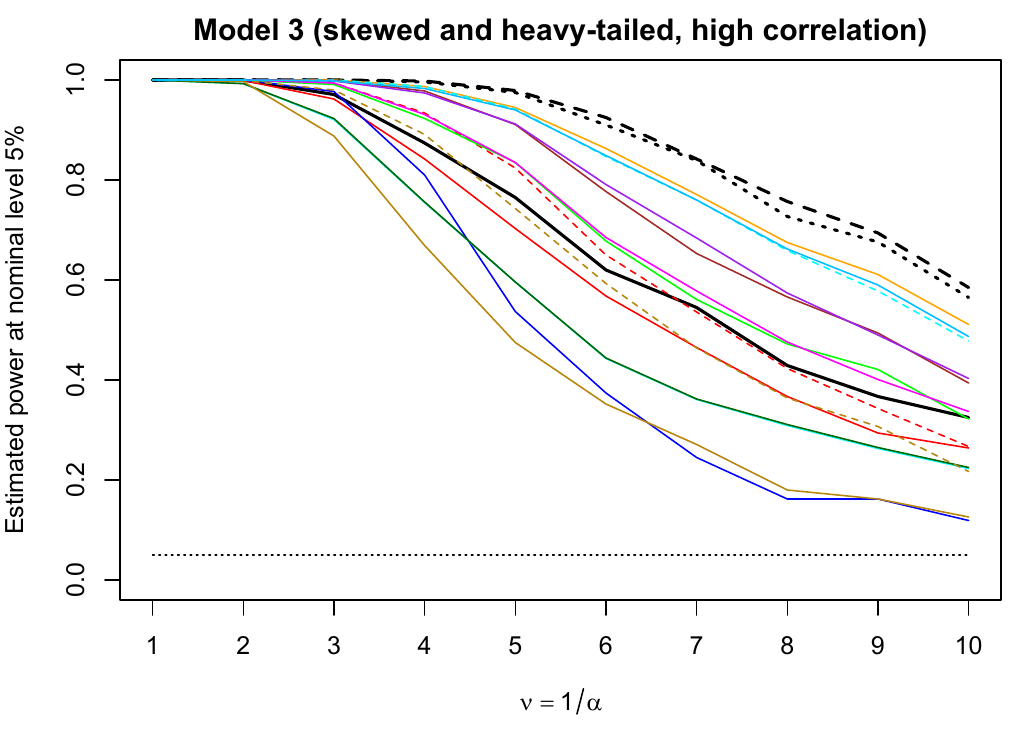}
\includegraphics[width=1\linewidth,trim={0 70 0 230},clip]{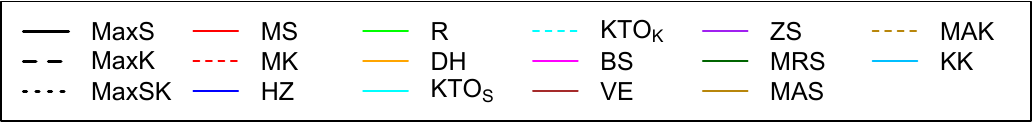}
\caption{Estimated powers of MaxS, MaxK, MaxSK, MS, MK, HZ, R, DH, $ \text{KTO}_\text{S} $, $ \text{KTO}_\text{K} $, BS, VE, ZS, MRS, MAS, MAK and KK tests for 5\% nominal level (horizontal dashed line) in Model 1 (top left), Model 2 (top right) and Model 3 (bottom) for $ n = 200 $, $p=5$, $q=2$ based on 1000 replicates. The horizontal dotted lines near the bottom of the plots correspond to the nominal level of 5\%.}
\label{fig:powerplotsn200eqcor99}
\end{figure}

In \autoref{fig:powerplotsn200eqcor99}, the estimated powers of the tests under high pairwise correlation display some notable differences compared to the estimated powers presented in Fig. 2 of the main paper in the case of the skewed model, namely, Model 1. In the plot for Model 1 with high correlation, it can be seen that the powers of the MaxS, MaxSK, R and MAS tests rise much more sharply compared to those in the corresponding plot in Fig. 2 of the main paper. Though the highest estimated power is displayed by the MaxS test, the powers of the R test and the MaxSK test are almost same with the MaxS test. The MAS test also displays high estimated power but it is lower than that of the MaxS test. In Model 2 and Model 3, the dominance of the MaxK and the MaxSK tests over all other tests do not change from those presented in Fig. 2 of the main paper. There, the estimated powers of the R test are lower than several other tests.

\begin{figure}[h!]
\centering
\includegraphics[height=5.9cm]{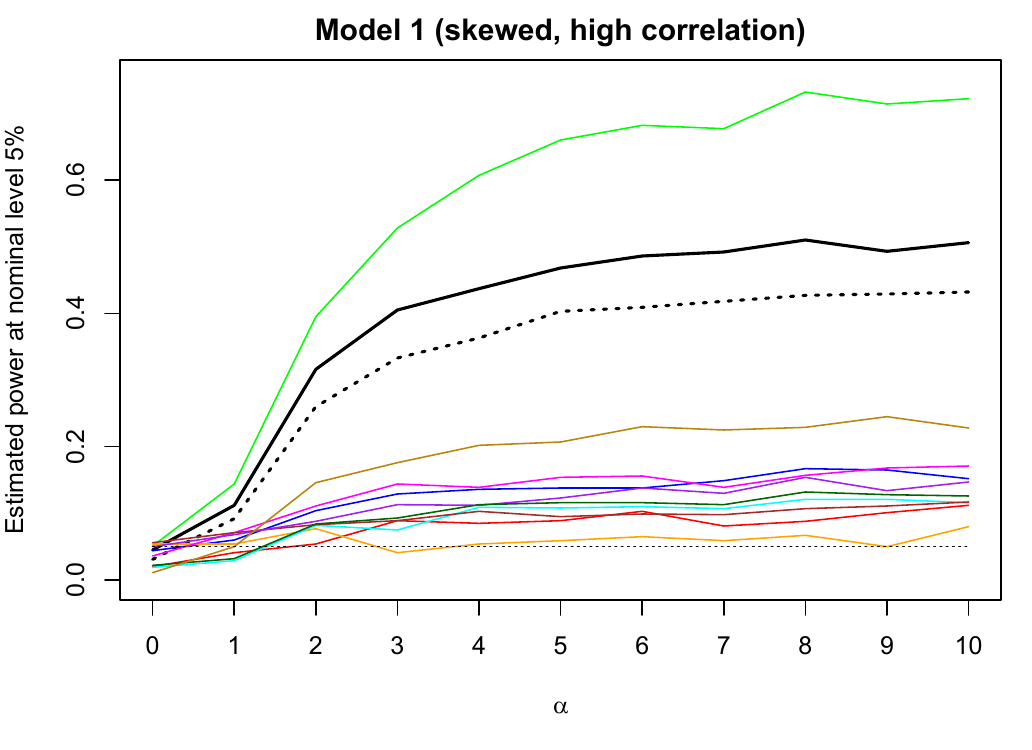}
\includegraphics[height=5.9cm]{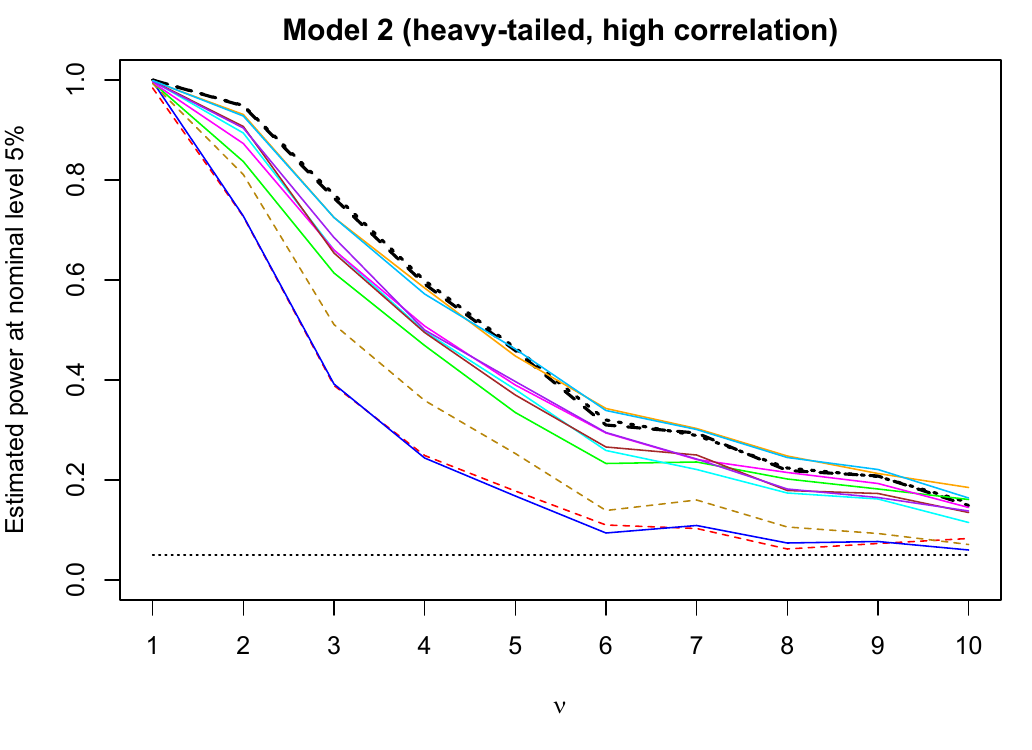}
\includegraphics[height=5.9cm]{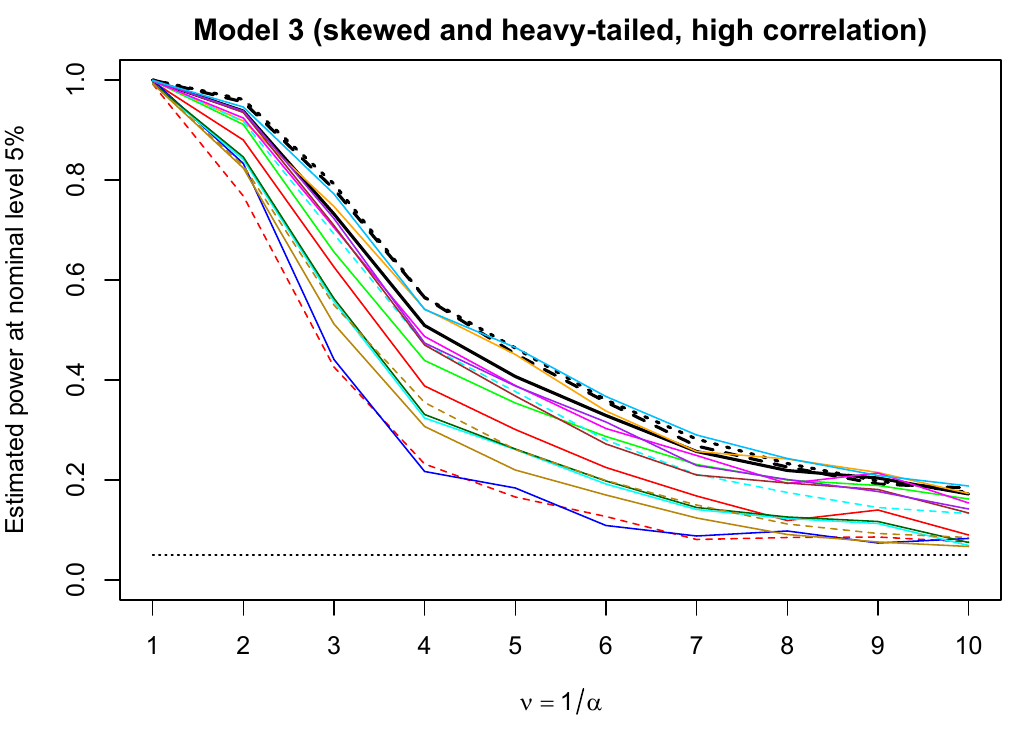}
\includegraphics[width=1\linewidth,trim={0 70 0 230},clip]{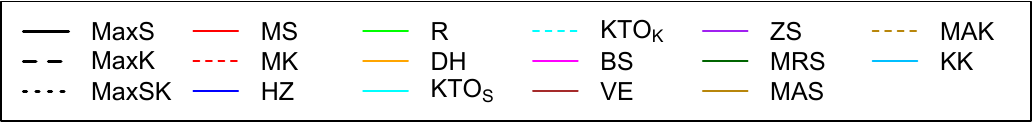}
\caption{Estimated powers of MaxS, MaxK, MaxSK, MS, MK, HZ, R, DH, $ \text{KTO}_\text{S} $, $ \text{KTO}_\text{K} $, BS, VE, ZS, MRS, MAS, MAK and KK tests for 5\% nominal level (horizontal dashed line) in Model 1 (top left), Model 2 (top right) and Model 3 (bottom) for $ n = 50 $, $p=5$, $q=2$ based on 1000 replicates. The horizontal dotted lines near the bottom of the plots correspond to the nominal level of 5\%.}
\label{fig:powerplotsn50eqcor99}
\end{figure}

In \autoref{fig:powerplotsn50eqcor99}, notable effects of the high pairwise correlation are visible on the estimated powers compared to those presented in \autoref{fig:powerplotsn50}.
In the skewed model, which is Model 1, the power of the R test is uniformly and significantly higher that all other tests, including the MaxS test and the MaxSK test, under high pairwise correlation. This is unlike what can be observed in \autoref{fig:powerplotsn50}, where the powers are estimated in a setup with moderate pairwise correlation. However, the performance of the R test deteriorates in Model 2 and Model 3 compared to the other tests. In those two models, the MaxK test and the MaxSK test perform better than the other tests, but the performances of a few other tests are close to them.